\documentclass[11pt,leqno]{article}
\usepackage{a4}
\usepackage{fullpage}
\usepackage{amsmath}
\usepackage{amsthm}
\usepackage{amssymb}
\usepackage{mathptmx}
\usepackage{epsfig}
\usepackage[numbers]{natbib}
\setcitestyle{numbers}
\def\poly{\mathrm{poly}}

\newcommand{\tarkhide}[1]{#1}
\newcommand{\tarkhideadd}[2]{#1}

\newcommand{\comment}[1]{}

\newcommand{\OMIT}[1]{} %

\def\literalqed{{\ \nolinebreak\hfill\mbox{\qedblob\quad}}}

\usepackage{ifthen}
\newcommand{\hugeDebug}{false}
\ifthenelse{\equal{\hugeDebug}{false}}{

}
{

}

\ifthenelse{\equal{\hugeDebug}{false}}{
\pagestyle{plain}
}
{
\pagestyle{headings}
}

\makeatletter
\newcommand{\singlespacing}{\let\CS=
\@currsize\renewcommand{\baselinestretch}{1}\tiny\CS}
\newcommand{\singlespacingplus}{\let\CS=
\@currsize\renewcommand{\baselinestretch}{1.25}\tiny\CS}
\newcommand{\doublespacing}{\let\CS=
\@currsize\renewcommand{\baselinestretch}{1.75}\tiny\CS}
\newcommand{\extradoublespacing}{\let\CS=
\@currsize\renewcommand{\baselinestretch}{1.9}\tiny\CS}
\newcommand{\draftspacing}{\let\CS=
\@currsize\renewcommand{\baselinestretch}{2.0}\tiny\CS}
\newcommand{\hugedraftspacing}{\let\CS=
\@currsize\renewcommand{\baselinestretch}{2.4}\tiny\CS}
\newcommand{\niceonespacing}{\let\CS=\@currsize\renewcommand{\baselinestretch}{1.1}\tiny\CS}
\newcommand{\nicetwospacing}{\let\CS=\@currsize\renewcommand{\baselinestretch}{1.2}\tiny\CS}
\newcommand{\nicethreespacing}{\let\CS=\@currsize\renewcommand{\baselinestretch}{1.3}\tiny\CS}
\newcommand{\nicehackspacing}{\let\CS=\@currsize\renewcommand{\baselinestretch}{1.334}\tiny\CS}
\newcommand{\singlespacingplusplus}{\let\CS=\@currsize\renewcommand{\baselinestretch}{1.35}\tiny\CS}
\newcommand{\nicefourspacing}{\let\CS=\@currsize\renewcommand{\baselinestretch}{1.4}\tiny\CS}
\newcommand{\nicefivespacing}{\let\CS=\@currsize\renewcommand{\baselinestretch}{1.5}\tiny\CS}
\newcommand{\nicesixspacing}{\let\CS=\@currsize\renewcommand{\baselinestretch}{1.6}\tiny\CS}
\newcommand{\nicesevenspacing}{\let\CS=\@currsize\renewcommand{\baselinestretch}{1.7}\tiny\CS}
\newcommand{\niceeightspacing}{\let\CS=\@currsize\renewcommand{\baselinestretch}{1.8}\tiny\CS}
\newcommand{\niceninespacing}{\let\CS=\@currsize\renewcommand{\baselinestretch}{1.9}\tiny\CS}
\makeatother

\makeatletter \@beginparpenalty=10000 \makeatother

\mathcode`\0="0030      %
\mathcode`\1="0031
\mathcode`\2="0032
\mathcode`\3="0033
\mathcode`\4="0034
\mathcode`\5="0035
\mathcode`\6="0036
\mathcode`\7="0037
\mathcode`\8="0038
\mathcode`\9="0039

\newcount\hour  \newcount\minutes  \hour=\time  \divide\hour by 60
\minutes=\hour  \multiply\minutes by -60  \advance\minutes by \time
\def\mmmddyyyy{\ifcase\month\or Jan\or Feb\or Mar\or Apr\or May\or Jun\or Jul\or
  Aug\or Sep\or Oct\or Nov\or Dec\fi \space\number\day, \number\year}
\def\hhmm{\ifnum\hour<10 0\fi\number\hour :%
  \ifnum\minutes<10 0\fi\number\minutes}
\makeatletter
\def\@cite#1#2{[#1\if@tempswa , #2\fi]}
\makeatother

\makeatletter
\def\@citex[#1]#2{\if@filesw\immediate\write\@auxout{\string\citation{#2}}\fi
  \def\@citea{}\@cite{\@for\@citeb:=#2\do
    {\@citea\def\@citea{,\linebreak[0]}\@ifundefined
       {b@\@citeb}{{\bf ?}\@warning
       {Citation `\@citeb' on page \thepage \space undefined}}%
\hbox{\csname b@\@citeb\endcsname}}}{#1}}
\makeatother

\makeatletter
\def\@cite#1#2{[#1\if@tempswa , #2\fi]}
\def\@citex[#1]#2{\if@filesw\immediate\write\@auxout{\string\citation{#2}}\fi
  \def\@citea{}\@cite{\@for\@citeb:=#2\do
    {\@citea\def\@citea{,\kern1pt\linebreak[0]}\@ifundefined
       {b@\@citeb}{{\bf ?}\@warning
       {Citation `\@citeb' on page \thepage \space undefined}}%
\hbox{\csname b@\@citeb\endcsname}}}{#1}}
\makeatother

\newcommand\qedblob{\mbox{$\Box$}}

\newcommand{\bigo}{{\protect\cal O}}

\newcommand{\naturals}{\mathbb{N}}

\newcommand{\score}{\mathrm{score}}

\newcommand{\calS}{{\cal S}}

\newtheorem{theorem}{Theorem}[section]

\newtheorem{corollary}[theorem]{Corollary}
\newtheorem{definition}[theorem]{Definition}
\newtheorem{lemma}[theorem]{Lemma}
\newtheorem{observation}[theorem]{Observation}

\newcommand{\p}{\ensuremath{\mathrm{P}}}
\newcommand{\np}{\ensuremath{\mathrm{NP}}}
\newcommand{\ccwm}{\ensuremath{\mathrm{CCWM}}}

\newcommand{\LEFT}{{{\mathrm{lt}}}}
\newcommand{\RIGHT}{{{\mathrm{rt}}}}

\DeclareMathOperator{\sord}{\mathit{L}}

\begin{document}

\title{The Complexity of Manipulative Attacks in Nearly Single-Peaked Electorates\thanks{Also appears as URCS-TR-2011-968.}}
\author{Piotr Faliszewski\\
        Department of Computer Science\\
        AGH Univ.\ of Science and Technology\\
        Krak\'ow, Poland
\and
        Edith Hemaspaandra\\
        Department of Computer Science\\
        Rochester Inst.\ of Technology\\
        Rochester, NY, USA
\and
        Lane A. Hemaspaandra\\
        Department of Computer Science\\
        Univ.\ of Rochester\\
        Rochester, NY, USA
}
\date{May 25, 2011, revised July 6, 2012}

\newcommand{\subjectcategories}{}

\newcommand{\whereareproofs}{
This paper touches on bribery, control, and manipulation, discusses
various election systems and notions of nearness to single-peaked, and
gives both polynomial-time attack results and NP-hardness results.  It
thus is not surprising that the proofs vary broadly in their
techniques and approaches; we have no single approach that covers this
entire range of cases. Almost all of our proofs are relegated
to the appendix.}

\newcommand{\whynodefinitions}{In this section 
we give intuitive descriptions of 
the problems that we study.  
More detailed coverage, and discussion of the
motivations and limitations of the models, can be found in the various
bibliography entries, 
including, for example, 
\citet{fal-hem-hem-rot:b-too-short:richer,fal-hem-hem:j:cacm-survey}.
We have also included formal definitions in the appendix.
As is standard, throughout this paper the 
terms ``NP-hard''/``NP-hardness'' will 
refer to polynomial-time many-one 
``NP-hard''/``NP-hardness.''}

\newcommand{\fixformissingproof}{In this version of this paper, the
  proof of Theorem~\ref{t:stcp} is located
  in Appendix~\ref{app:stcp}, and we assume that the
  interested reader will now read it (doing so is not required and
  probably not even a good idea on a first reading, but such reading
  will make clearer the next paragraph, which briefly refers to the
  algorithm within that proof). }

\newcommand{\acknowledgements}{%
\paragraph*{Acknowledgments}
We thank the anonymous referees for their comments.
Piotr Faliszewski was supported in
part by AGH University of Technology grant 11.11.120.865, by
Polish Ministry of Science and Higher Education grant N-N206-378637,
and by the Foundation for Polish Science's Homing/Powroty program.  Edith
Hemaspaandra was supported in part by grant NSF-IIS-0713061 and a
Friedrich Wilhelm Bessel Research Award.  Lane A. Hemaspaandra was
supported in part by grants NSF-CCF-0915792 and ARC-DP110101792, and a
Friedrich Wilhelm Bessel Research Award.  }

\maketitle

\begin{abstract}
  Many electoral bribery, control, and manipulation problems (which we
  will refer to in general as ``manipulative actions'' problems) are
  $\np$-hard in the general case.  It has recently been noted that many
  of these problems fall into polynomial time if the electorate is
  single-peaked (i.e., is polarized along some axis/issue).  However,
  real-world electorates are not truly single-peaked.  There are
  usually some mavericks, and so real-world electorates tend to merely
  be nearly single-peaked.  This paper studies the complexity of
  manipulative-action algorithms for elections over nearly
  single-peaked electorates, for various notions of nearness and
  various election systems.  We provide instances where even one
  maverick jumps the manipulative-action complexity up to $\np$-hardness,
  but we also provide many instances where a reasonable number 
  of mavericks can be tolerated without
  increasing the manipulative-action complexity.
\end{abstract}

\subjectcategories

\section{Introduction}\label{sec:intro}
Elections are a model of collective 
decision-making so central in human and multiagent-systems contexts---ranging 
from planning to collaborative filtering to reducing 
web spam---that it is natural to want to get a handle on the 
computational difficulty 
of finding whether manipulative actions can obtain a given outcome
(see the survey~\citep{fal-hem-hem:j:cacm-survey}).
A recent line of work 
\tarkhide{started by
Walsh }%
\citep{wal:c:uncertainty-in-preference-elicitation-aggregation,fal-hem-hem-rot:cOUTDATEDbyJOURNALwithPTR:single-peaked-preferences,bra-bri-hem-hem:c:sp2}
has looked at the extent to which $\np$-hardness results for the
complexity of manipulative actions (bribery, control, and
manipulation) may evaporate when one focuses on electorates that are
(unidimensional) 
single-peaked, a central social-science model of electoral behavior.
That model basically views society as polarized along some (perhaps
hidden) issue or axis.  However, real-world elections are unlikely to
be perfectly single-peaked.  Rather, they are merely very close to being
single-peaked, a notion that was recently raised in a computational
context by \citet{con:j:eliciting-singlepeaked}
and 
\citet{esc-lan-ozt:c:single-peaked-consistency}.  
There
will almost always be a few mavericks, who vote based on some reason
having nothing to do with the societal axis.  For example, in recent US
presidential primary and final elections, there was much discussion of
whether some voters would vote not based on the political positioning
of the candidates but rather based on the candidates' religion, race, or
gender.  In this paper, we most centrally study whether the
evaporation of complexity results that often holds 
for 
single-peaked electorates will also occur in nearly single-peaked
electorates.  We prove that often the answer is yes, and sometimes
the answer is no. 
We defer to Section~\ref{sec:related} our
discussion of previous and related work.

Among the contributions
of our paper are the following.

\begin{itemize}
\item Most centrally, we show that in many control and bribery settings,
a reasonable number of mavericks (voters whose votes are not
consistent with the societal axis) can be 
handled.
In such cases,
the ``complexity-shield evaporation'' results of the earlier work can
now be declared free from the worry that the results might hold only for
perfect single-peakedness.

\item We give settings, 
for example 3-candidate Borda and 3-candidate veto, in 
which even one maverick raises the 
(constructive 
coalition weighted)
manipulation complexity from $\p$ to $\np$-hardness.

\item For all scoring systems of the form $(\alpha_1,\alpha_2,\alpha_3)$,
$\alpha_2 \neq \alpha_3$, we provide a dichotomy theorem 
determining when the (constructive coalition 
weighted) manipulation problem 
is in $\p$ and when it is $\np$-complete, for 
so-called 
single-caved societies.

\item We show cases where the price of mavericity is 
paid in nondeterminism---cases where, for each $k$, 
we prove the control problem for 
societies with 
$\bigo(\log^kn)$ mavericks 
to be in complexity class $\beta_k$, the $k$th level of the limited
nondeterminism hierarchy 
of
\citet{fis-kin:j:beta}.
\end{itemize}

\whereareproofs

\section{Preliminaries}\label{sec:prelims}
\whynodefinitions

\smallskip

{\bf Elections}\quad An election $E=(C,V)$ consists of a finite
candidate set $C$ and a finite collection $V$ of votes over the candidates.  
$V$ is a list of entries, one per voter, 
with each entry containing a linear (i.e., 
tie-free total) ordering of the candidates
(except for approval elections where each vote is a $\|C\|$-long 
0-1 
vector
denoting disapproval/approval of each
candidate).\footnote{Throughout this paper $V$, though input as a list
(one ballot
per voter), typically functions as a multiset.  When dealing with voter
sets, we use terms such as set/subset to mean multiset/submultiset
(although we will typically just say collection), and we use set-bracket
notation to, for that case, mean multisets, e.g., if $v$ is a preference
order, $w$ is a preference order, $V=\{v,v\}$, and $W=\{v,w\}$, then $V
\cup W = \{v,v,v,w\}$.}
In plurality elections,
whichever candidate gets the most top-of-the-preference-order votes
wins.  Each vector $(\alpha_1,\ldots,\alpha_k)$, $\alpha_i \in
\naturals$, $\alpha_1 \geq \cdots \geq \alpha_k \geq 0$, defines a
$k$-candidate scoring protocol election, in which each voter's $i$th
favorite candidate gets $\alpha_i$ points, and whichever candidate
gets the most points wins.  $k$-candidate veto is defined by the
vector $(\overbrace{1,\ldots,1}^{k-1},0)$, and $k$-candidate Borda 
is
defined by the vector $(k-1,k-2,\ldots,0)$.  In approval elections,
whichever candidate is approved of by the most voters wins.  In all
the systems just mentioned, if candidates tie for the highest number
of points, those tieing for highest are all considered winners.  In
Condorcet elections, a candidate wins if he or she strictly beats
every other candidate in pairwise head-on-head votes.

\smallskip

{\bf Attacks}\quad 
Each of the election problems is defined based
on an election and some additional parameters.  In constructive
coalition weighted manipulation (CCWM), the input is a set of
nonmanipulative voters (having weights and preferences over the
candidates),
a list of the weights of the manipulative voters, and which candidate
$p$ the manipulators wish to be a winner.  (Input) instances are in
the set exactly if there is a set of votes the manipulators can cast
to make $p$ a winner (under the given election system).  In
``bribery,'' all voters have preferences, and our input is an election,
a candidate $p$, and a bound $K$ on how many voters can be bribed.
Instances are in the set if there is a way of changing the preferences
of at most $K$ voters that 
makes 
$p$ a winner.  (We will briefly mention a
number of variations of bribery.
``Weighted'' means the voters have weights, ``\$'' means voters have
individual prices (and $K$ becomes a bound on the amount that can be
spent seeking to make $p$ win).
For approval voting, ``negative'' bribery
means a bribe cannot change someone from disapproving of $p$ to
approving of $p$, and ``strongnegative'' bribery means
every bribed person must
end up disapproving of $p$.  In negative bribery, one can only help
$p$ in subtle, indirect ways.  For plurality, the negative notion 
is similar, see
\citet{fal-hem-hem:j:bribery} or our appendix.)  
In control (of the four types we will discuss), the input
is an election, the candidate $p$ one 
wants to be a
winner, and a parameter $K$ limiting how many actors one can influence
in the designated way.  Our four types of control will be adding voters
(CCAV), deleting voters (CCDV), adding candidates (CCAC), and deleting
candidates (CCDC).  Each of those four problems is defined as 
the collection of inputs on which using at most $K$ actions of the 
designated type (e.g., adding at most $K$ voters) suffices to make $p$ a 
winner.
For CCAV an additional part of the input is a pool of potential 
additional voters (and their preferences over the candidates).
For CCAC an additional part of the input is a pool of potential 
additional candidates, and all voters have preferences over the 
set of all initial and potential-additional candidates.
(Some early papers on control 
focused on making $p$ be the one and only winner,
but we follow the more recent approach of focusing on making $p$ 
become a winner.  
The older results for the former case that we 
cite here are known in the literature, 
or were verified for this paper by us, to 
also hold for the latter case.)
Control loosely models such real-world activities as get-out-the-vote
drives, targeted advertising, and voter suppression.

Each of our algorithms for manipulation, bribery, and control not only
gives a \emph{yes/no} answer for the decision variant of the
problem, but also can be made to produce a successful
manipulative action if the answer is \emph{yes}.

Having algorithms for such tasks as bribery and control isn't
inherently a ``bad'' or unethical thing.  For example, bribery and
control algorithms are valuable tools for actors (party chairs,
campaign managers, etc.) who are trying to most effectively use their
resources.

\smallskip

{\bf (Nearly) Single-peakedness}\quad A collection $V$ of votes (cast
as linear orders) is said to be single-peaked
exactly if there is a
linear order $L$ over the candidate set such that for each triple of
candidates, $c_1$, $c_2$, $c_3$, it holds that if $ c_1 \, L \, c_2 \,
L c_3 \lor c_3 \, L \, c_2 \, L \, c_1$, then 
$(\forall v \in V)[ c_1 \,P_v\,
c_2 \implies c_2 \,P_v \,c_3]$, where $a P_v b$ means that 
voter $v$ prefers $a$ to $b$.  This notion was first created by
Black more than half a century ago, and is one of the most important
concepts in political science.  Loosely put, the notion is motivated
as follows.  Imagine that on some issue, for example what the tax
rate should be for the richest Americans, each person has a utility
curve that on a (perhaps empty) initial part is nondecreasing and then
on the (perhaps empty) rest is nonincreasing.  Suppose the candidates
are spread along the tax-rate axis as to their positions, 
with no two on
top of each other.  
The set of preferences that can be supported among
them by curves of the mentioned sort on which there are no ties among
candidates in utility are precisely the single-peaked vote
ensembles. Note that different voters can have different
peaks/plateaus and different curves, e.g., if both Alice and Bob think
40 percent is the ideal top tax rate, it is completely legal for Alice
to prefer 30 percent to 50 percent and Bob to prefer 50 percent to 30
percent.  
There is extensive political science literature on single-peaked
voting's naturalness, ranging from 
conceptual discussions to 
empirical
studies of actual 
US 
political elections (with few candidates) 
showing
that most voters are single-peaked with respect to left-right
political spectrum, and has been described as
``\emph{the} canonical setting for models of political        
institutions''~\citep{gai-pat-pen:b:arrow-on-single-peaked-domains}.
If in the definition of single-peaked one 
replaces the final ``forall'' with this one,
$(\forall v \in V)[ c_2 \,
P_v \,
c_1 \implies c_3 \, P_v \, c_2]$, one defines the closely
related notion of single-caved preferences, which we 
will also study.
For approval ballots, a vote set $V$ is said to be 
single-peaked if there is a linear order $L$ 
such that for each voter $v$, all candidates 
that $v$ approves of (if any) form an adjacent block in $L$.

In all our manipulative action problems about single-peaked and nearly
single-peaked societies, we will follow Walsh's model, which is that
the societal order, $L$, is part of the input.  (See the earlier
papers for extensive discussion of why this is a reasonable model.)

In this paper, we will primarily focus on elections whose voters are
``nearly'' single-peaked, under the following notions of nearness.
Our ``maverick'' notions apply to both voting by approval 
ballots and voting by linear orders;  our other notions are 
specific to voting by linear orders.
We will say an election is over a $k$-maverick-SP society
(equivalently, a $k$-maverick-SP electorate) if all but $k$ of the
voters are consistent with (in the sense of single-peakedness; this does
not mean identical to) the societal order $L$.  That is, we allow up to $k$
mavericks.  
We will speak of $f(\cdot)$-maverick-SP societies when this usage and
the type of $f$'s argument(s) is clear from context ($f$'s argument(s)
will typically be the size of the election instance or some parameters
of the election, e.g., the number of candidates or the number of
voters).

Also, we will prove a number of results 
that state that ``PROBLEM for ELECTION-SYSTEM 
over log-maverick-SP societies is 
in $\p$''; this is a shorthand for the claim that 
for each
function $f$ 
(that is computable in time polynomial in 
the size of the input---which is roughly 
$\|V\| \|C\| \log \|C\|$ for the election
$(C,V)$
itself plus whatever space is taken by
other parameters---%
and to avoid possible technical problems, we should assume 
$f$ is nondecreasing)
whose 
\emph{value} is $\bigo(\log(\mathrm{ProblemInputSize}))$,
it holds that 
``PROBLEM
for ELECTION-SYSTEM over $f$-maverick-SP societies is in $\p$,''
where the argument to $f$ is the input size of the problem.
An election is over a $(k,k')$-swoon-SP society if each voter
has the property that if one removes the voter's $k$ favorite
and $k'$ least favorite
candidates from the voter's preference order, the resulting 
order is consistent
with societal order $L$ after removing those same candidates from $L$.
We will use swoon-SP as a shorthand for 
\tarkhideadd{$(1,0)$-swoon-SP,
as we will not study other swoon values in this paper.}{$(1,0)$-swoon-SP\@.}
In swoon-SP, each person may have as her or his
favorite some candidate chosen due to some personal passion (such as
hairstyle or religion), but all the rest of
that person's vote must be consistent with the societal polarization.
An election is over a 
Dodgson$_k$-SP society if for each 
voter some 
at-most-$k$ sequential exchanges of adjacent candidates in his or her
order make the vote consistent with the societal order $L$.
An election is over a PerceptionFlip$_k$-SP society if, for each 
voter, there is some series of 
at most $k$ sequential exchanges of adjacent candidates
\emph{in the societal order $L$}
after which the voter's 
vote is 
consistent with $L$.  
This models each voter
being consistent with that voter's humanly blurred view of 
the societal order.

Some model details follow.  For control by adding voters for
maverick-SP societies, the total number of mavericks in the initial
voter set and the pool of potential additional voters is what the
maverick bound limits.  
For manipulation, nonmanipulators as well as manipulators
can be mavericks, and we bound the total number of mavericks.
For bribery involving $f(\cdot)$-maverick-SP societies, we will
consider both the ``standard'' model and the ``marked'' model.  In the
standard model, the $f(\mathrm{ProblemInputSize})$ 
limit on the number of mavericks must hold both
for the input and for the voter set after the bribing is done; anyone
may be bribed and bribes can create mavericks and can make mavericks
become nonmavericks.  In the marked model, 
each voter has a flag
saying whether or not he or she can cast a
maverick vote (we will call that being
``maverick-enabled'').
The at most $f(\mathrm{ProblemInputSize})$ voters with 
the maverick-enabled flag may
(subject to the other constraints of the bribery problem such as
total number of bribes) be bribed in any way, and so may legally cross in
either direction between consistency and inconsistency with the
societal ordering.  All non-maverick-enabled voters must be consistent
with societal order $L$ both before and after the bribing, although
they too can be bribed (again, subject to the problem's other constraints
such as total number of bribes).

For single-peaked electorates, ``median voting'' (in which the candidate wins 
who on
the societal axis is preferred by the ``median voter'') is known to be
strategy-proof, i.e., a voter never benefits from misrepresenting his 
or her preferences.  It might seem tempting to conclude from that 
that all elections on single-peaked societies
``should'' use median voting, and that we thus need not discuss
single-peaked (or perhaps even nearly single-peaked) elections 
with respect to other voting systems, such as plurality, veto, etc.
But that temptation should be resisted.
First, median voting's 
strategy-proofness regards manipulation, not control or
bribery.  Second, even in real-world political elections broadly
viewed as being (nearly) single-peaked, it simply is not the case that
median voting is used.  People,
for 
whatever reasons of history and comfort, 
use such systems as plurality, approval, and so on for such elections.  
And so algorithms
for those systems are worth studying. Third, for 
manipulation of nearly single-peaked electorates, strategy-proofness does not 
even hold.  And although for them indeed only the mavericks can have 
an incentive to lie, 
that doesn't mean that the outcome
won't be utterly distorted even by 
a single
maverick.  
There are arbitrarily large electorates, having just one
maverick, where that maverick can change the winner from being the
median one to instead being a candidate on the 
outer extreme
of the societal order.  

\section{Manipulation}\label{sec:manip}
This paper's sections on control and bribery focus on, and provide
many examples of, settings where not just the single-peaked
case but even the nearly single-peaked cases have polynomial-time
algorithms.  Regarding manipulation, the results are more sharply
varied.  

We show that $\np$-hardness holds for a rich class of 
scoring protocols, in the presence of even one maverick.
(When $\alpha_2 = \alpha_3$ the system is either equivalent to plurality
or is a trivial system where everyone always is a winner.  These cases
are easily seen to be in $\p$.)
Recall from Section~\ref{sec:prelims} 
the meaning of ``$(\alpha_1,\alpha_2,\alpha_3)$ elections,'' namely,
scoring protocol elections
using the vector $(\alpha_1,\alpha_2,\alpha_3)$.

\begin{theorem}\label{t:1msp-dich}
  For each $\alpha_1 \geq \alpha_2 > \alpha_3$, $\ccwm$ for
  $(\alpha_1,\alpha_2,\alpha_3)$ elections over 1-maverick-SP societies is
  $\np$-complete.
\end{theorem}

\tarkhide{
We point out that this theorem is of the same form
as that for the general case 
(see~\citep{con-lan-san:j:when-hard-to-manipulate,hem-hem:j:dichotomy,pro-ros:j:juntas}).  
However, the proofs for 
the general case do not work in our case,
since those proofs construct elections with at least two mavericks. 
}

In the general case (i.e., no single-peakedness is required), the
above cases also are $\np$-complete~\citep{con-lan-san:j:when-hard-to-manipulate,hem-hem:j:dichotomy,pro-ros:j:juntas}, so allowing a one-maverick
single-peaked society is jumping us up to the same level of complexity
that holds in the general case here.  In contrast, for SP societies
(without mavericks), 
3-candidate $\ccwm$ is
$\np$-complete when $(\alpha_1 - \alpha_3) > 2(\alpha_2 - \alpha_3) > 0$ 
and is in $\p$ 
otherwise~\citep{fal-hem-hem-rot:j:single-peaked-preferences}.  So, in
particular, 3-candidate veto and 3-candidate Borda elections are in $\p$
for the SP (single-peaked) 
case, but are already $\np$-complete for SP with one maverick
allowed. 

Does allowing one maverick always raise the $\ccwm$ complexity?
No, as the following
theorem shows.
\tarkhide{(The $k = 0$ case follows from
\citet{fal-hem-hem-rot:j:single-peaked-preferences}.)}

\begin{theorem}\label{t:veto-easy}
For each $k \geq 0$ and $m \geq k + 3$, $\ccwm$ for
$m$-candidate veto elections over $k$-maverick-SP 
societies is in $\p$.
\end{theorem}

In contrast, all of Theorem~\ref{t:veto-easy}'s cases are well-known to be
$\np$-complete in the general 
case~\citep{con-lan-san:j:when-hard-to-manipulate}.
Still, the 
contrast is a bit fragile.  For example, although the above
theorem shows that $\ccwm$ for $(1,1,1,1,0)$ elections over
$2$-maverick-SP societies is in $\p$,
we prove below that $\ccwm$ for
$(1,1,1,0)$ elections over $2$-maverick-SP societies is 
$\np$-complete. 
Note also that this theorem gives an example where
4-candidate veto elections are $\np$-complete but 5-candidate veto
elections are in $\p$, in contrast with the behavior that one often
expects regarding $\np$-completeness and parameters, namely, one might
expect that increasing the number of candidates wouldn't lower the
complexity.  (However, see \citet{fal-hem-hem-rot:j:single-peaked-preferences} for another
example of this unusual behavior.)

\begin{theorem}\label{t:veto-hard}
For each $k \geq 0$ and $m \geq 3$ such that $m \leq k + 2$, $\ccwm$ for
$m$-candidate veto elections over $k$-maverick-SP societies 
is $\np$-complete.
\end{theorem}

Theorems~\ref{t:veto-easy} and~\ref{t:veto-hard} also show that
for any number of mavericks, there exists a voting system
such that $\ccwm$ is easy for up to that number of mavericks,
and hard for more mavericks.

\begin{corollary}\label{t:veto-cor}
Let $k \geq 0$.   For all $\ell \geq 0$, 
$\ccwm$ for $k+3$-candidate veto elections over $\ell$-maverick-SP 
societies is in $\p$ if $\ell \leq k$ and $\np$-complete otherwise.
\end{corollary}

Let us now turn from the maverick notion of nearness to single-peakedness,
and look at the ``swoon'' notion\tarkhideadd{, in which, recall, each voter must be
consistent with the societal ordering (minus the voter's first-choice
candidate) when one removes from the voter's preference list the
voter's first-choice candidate.}{.}  We will still mostly focus on the
case of veto elections.  
For three candidates (see Observation~\ref{o:three-cand})
and four candidates
we have $\np$-completeness, and for five or more candidates we 
have membership in $\p$.
\begin{theorem}\label{t:swoon-ccwm-veto}
  For each $m \geq 5$, $\ccwm$ for
  $m$-candidate veto elections
  in swoon-SP societies is in $\p$\@.
  For $m \in \{3,4\}$, this problem is $\np$-complete.
\end{theorem}

\begin{observation}\label{o:three-cand}
Every 3-candidate election is a swoon-SP election and
a Dodgson$_1$-SP election and so all complexity results for 3-candidate
elections in the general case also hold for 
swoon-SP elections and Dodgson$_1$-SP elections. 
\end{observation}

Complexity results for general elections do not always hold for swoon-SP 
elections or for Dodgson$_1$-SP elections.  For example, for $m \geq 5$,
$\ccwm$ for $m$-veto elections is $\np$-complete in the general case,
but in $\p$ for swoon-SP societies (Theorem~\ref{t:swoon-ccwm-veto})
and Dodgson$_1$-SP societies (Theorem~\ref{t:dodgson-ccwm-veto}).

\begin{theorem}\label{t:dodgson-ccwm-veto}
  For each $m \geq 5$, $\ccwm$ for
  $m$-candidate veto elections
  in Dodgson$_1$-SP societies is in $\p$\@.
  For $m \in \{3,4\}$, this problem is $\np$-complete.
\end{theorem}

We conclude this section with a brief comment about single-caved
electorates.  (We remind the reader that single-caved is not a
``nearness to SP'' notion, but rather is in some sense a
mirror-sibling.)  For scoring vectors $(\alpha_1,\alpha_2,\alpha_3)$,
the known $\ccwm$ dichotomy result for single-peaked electorates 
is that if $\alpha_1 - \alpha_3 > 2(\alpha_2 - \alpha_3)$ 
then the problem is $\np$-complete
and otherwise the problem is in $\p$.  For single-caved, 
the opposite holds for each case that is not in P
in the general case.

\begin{theorem}\label{t:manip-sc}
  For each $\alpha_1 \geq \alpha_2 > \alpha_3$, $\ccwm$ for
  $(\alpha_1,\alpha_2,\alpha_3)$ elections over single-caved societies is
  $\np$-complete if 
  $(\alpha_1 - \alpha_3) \leq 2(\alpha_2-\alpha_3)$  
  and otherwise is in $\p$.
\end{theorem}

\section{Control}\label{sec:control}
The very first results 
of \citet{fal-hem-hem-rot:j:single-peaked-preferences} showing
that $\np$-complete general-case control results can simplify to $\p$ results
for single-peaked electorates were for constructive control by adding
voters and for constructive control by deleting voters, for approval
elections.
We show that each of those results can be 
reestablished even in the presence of 
logarithmically many mavericks.  (Indeed, we mention in passing
that even if the attacker is allowed to simultaneously both
add and delete voters---so-called ``AV+DV'' multimode 
control in the recent
model that allows simultaneous 
attacks~\citep{fal-hem-hem:c:multimode}---the complexity of 
planning an optimal attack still remains polynomial-time 
even with logarithmically many mavericks.)  
\begin{theorem}\label{t:stcp}
CCAV and CCDV for approval
elections over log-maverick-SP societies are each 
in $\p$.  For CCAV, the complexity remains in $\p$
even for the case where no limit is imposed on the 
number of mavericks in the initial voter set, and the
number of mavericks in the set of potential
additional voters is logarithmically 
bounded (in the overall problem input size).\footnote{By 
that last part, we mean 
precisely the definition---including its 
various restrictions on the complexity of the function---of the notion
of log-maverick-SP, except with the limit being placed just 
on the number of mavericks in the additional voter set.
\tarkhide{Formally put, to avoid any possibility of ambiguity or 
confusion, we mean that 
for each
function $f$ 
(that is computable in time polynomial in 
the size of the input;
and to avoid possible technical problems, we require that 
$f$ be nondecreasing)
whose 
\emph{value} is $\bigo(\log(\mathrm{ProblemInputSize}))$,
it holds that 
``CCAV
for approval elections over inputs for which (as is standard
in our model regarding SP and nearly-SP cases, a societal 
ordering of single-peakedness is given a part of the input, and)
at most $f(\mathrm{ProblemInputSize})$ of the additional 
voters are inconsistent with the societal ordering
is in $\p$.''}}%
\end{theorem}

\tarkhide{Although we will soon prove a more general result, we start by proving
this result directly.  We do so both as that will make the more
general result clearer, and as P-time results are the core focus of
this paper.}  Our proof involves the ``demaverickification'' of the
society, in order to allow us to exploit the power of
single-peakedness.  
\tarkhide{By doing so, our proof establishes that there is a 
polynomial-time 
disjunctive truth-table reduction (see \citet{lad-lyn-sel:j:com} for
the formal definition of $\leq_{dtt}^p$, but one does not need to know
that to follow our proof) to the single-peaked case.}
\fixformissingproof
\tarkhideadd{Now, before we move on to other election systems, let us
  pause to wonder whether Theorem~\ref{t:stcp} is just the tip of an
  iceberg, and is in fact hiding some broader connection between
  number of mavericks and computational complexity theory.}{Before 
  we move on to other election systems, let us pause to look at 
  whether Theorem~\ref{t:stcp} is a reflection of some
  broader connection between number of mavericks and computational
  complexity theory.}
We won't give this type of discussion
for all, or even most, of our theorems.  But it is worthwhile to,
since this is our first control result, look here at what 
holds.  \tarkhideadd{And
what holds is that Theorem~\ref{t:stcp} is indeed in some sense 
the tip of an iceberg.  However, it is an iceberg whose tip is  
its most interesting part, since it gives the part that admits
polynomial-time attacks.  

Still, the rest of the iceberg 
brings out an interesting connection between maverick 
frequency and nondeterminism.
Let us think again of the proof we just
saw.  It worked by sequentially generating each member of the powerset
of a logarithmic-sized set (call it $Q$).  And we did that, naturally
enough, in polynomial time.  However, note that we could also have
done it with nondeterminism.  We can nondeterministically guess for
each member of $Q$ whether or not it will be added (for CCAV) or
deleted (for CCDV).  And then after that nondeterministic guess, we
for the CCAV case do the demaverickification presented in the above
theorem's proof and for the CCDV case do the deletable/nondeletable
marking, and then we run the polynomial-time algorithms for the
single-peaked approval-voting CCAV and the approval-voting CCDV (with
deletable/nondeletable flag, and all mavericks---those not consistent
with the societal ordering---being nondeletable) cases.  It is easy to
see that above proof argument works fine with the change to
nondeterminism.  Indeed, the reason Theorem~\ref{t:stcp} is 
about ``$\p$'' is because sequentially handling $\bigo(\log(\mathrm{ProblemInputSize}))$
nondeterministic bits can be done in polynomial time.

So, what }{What }%
\newcommand{\NONDETTIME}{{\rm NONDET\mbox{-}TIME}}%
underlies the above theorem are the following results that say
that frequency of mavericks in one's society exacts a price, in
nondeterminism.  (We here are 
proving just an upper bound, but we
conjecture that the connection is quite tight---that commonality of
wild voter behavior is very closely connected with nondeterminism.)  To
state the results, we need to briefly introduce some notions from
complexity theory.  Complexity theorists often separate out the
weighing of differing resources, putting bounds on each.  The only 
such class we need here is the class of languages that can be 
accepted in time $t(n)$ on machines using $g(n)$ bits of 
nondeterminism, which is typically denoted 
$\NONDETTIME[ g(n) , \, t(n)]$.  The most widely known
such classes are those of the limited nondeterminism 
hierarchy, known as the beta hierarchy, 
of \citet{fis-kin:j:beta} (see also~\citep{dia-tor:j:beta} 
and the survey~\citep{gol-lev-mun:j:limited-nondeterminism}).  
$\beta_k$ is the class
of sets that can be accepted in polynomial time on machines that use
$\bigo(\log^kn)$ bits of nondeterminism: $\beta_k = \{ L ~|~ (\exists
\mbox{ polynomial } t(n)) (\exists g(n) \in \bigo(\log^kn)) [L \in
\NONDETTIME[ g(n),\, t(n)]]$, or for short, $\beta_k = \NONDETTIME[ 
 \bigo(\log^kn),\, \poly]$.  Of course, $\beta_0 = \beta_1 = \p$\@.
We can now state our result, which
says that frequency of mavericity is paid for in nondeterminism.

\begin{theorem}\label{t:stcp-general}
CCAV and CCDV for approval
elections over $f(\cdot)$-maverick-SP societies are each 
in $\NONDETTIME[f(\mathrm{ProblemInputSize}),\, \poly]$.
For CCAV, the complexity remains in 
$\NONDETTIME[f(\mathrm{ProblemInputSize}),\, \poly]$
even for the case where no limit is imposed on the 
number of mavericks in the initial voter set, and the
number of mavericks in the set of potential
additional voters is $f(\cdot)$-bounded 
(in the overall problem input size).
\end{theorem}

\begin{corollary}\label{c:stcp-beta}
For each natural number $k$,
CCAV and CCDV for approval
elections over $\bigo(\log^kn)$-maverick-SP societies are each 
in $\beta_k$.
For CCAV, the complexity remains in 
$\beta_k$
even for the case where no limit is imposed on the 
number of mavericks in the initial voter set, and the
number of mavericks in the set of potential
additional voters is $\bigo(\log^kn)$-bounded 
(in the overall problem input size).
\end{corollary}

Theorem~\ref{t:stcp} follows from the $k=1$ case of this more general
corollary.  Now let us turn from our particularly detailed discussion
of CCAV and CCDV for approval voting, and let us look at other
election systems.

For Condorcet elections, both CCAV and CCDV are known to be 
$\np$-complete in the general case~\citep{bar-tov-tri:j:control} but 
to be in $\p$ for single-peaked electorates~\citep{bra-bri-hem-hem:c:sp2}.
For Condorcet, results analogous to those for approval under
$f(\cdot)$-maverick-SP societies hold.

\begin{theorem}\label{t:stcp-condorcet}
CCAV and CCDV for Condorcet
elections over $f(\cdot)$-maverick-SP societies are each 
in $\NONDETTIME[f(\mathrm{ProblemInputSize}),\, \poly]$.
For CCAV, the complexity remains in 
$\NONDETTIME[f(\mathrm{ProblemInputSize}),\, \poly]$
even for the case where no limit is imposed on the 
number of mavericks in the initial voter set, and the
number of mavericks in the set of potential
additional voters is $f(\cdot)$-bounded 
(in the overall problem input size).
\end{theorem}

\begin{corollary}\label{c:stcp-beta-condorcet}
For each natural number $k$,
CCAV and CCDV for Condorcet
elections over $\bigo(\log^kn)$-maverick-SP societies are each 
in $\beta_k$.
For CCAV, the complexity remains in 
$\beta_k$
even for the case where no limit is imposed on the 
number of mavericks in the initial voter set, and the
number of mavericks in the set of potential
additional voters is $\bigo(\log^kn)$-bounded 
(in the overall problem input size).
\end{corollary}

For plurality, the most important of systems, 
CCAV and CCDV are known to be in $\p$ for the general 
case~\citep{bar-tov-tri:j:control}, so there is no need to
seek a maverick result there.  However,
CCAC and CCDC are both known to be $\np$-complete in
the general case and in $\p$ in the single-peaked 
case.  For both of those, a constant number of mavericks can be handled.
\begin{theorem}\label{t:plurality-maverick}
For each $k$,
CCAC and CCDC for plurality
over $k$-maverick-SP societies are
  in $\p$.
\end{theorem}
\tarkhide{As in the case of Theorem~\ref{t:stcp}, the idea of the
  algorithm is
  to (polynomial-time disjunctively truth-table) 
  reduce to the single-peaked case. However, here the
  mavericks require a more involved brute-force search and thus we can
  only handle a constant number of them.
}

Unfortunately, swooning cannot be
handled at all (unless $\p=\np$, of course).
Also, allowing the number of mavericks to be some
root of the input size cannot be handled.
\begin{theorem}\label{t:plurality-swoon}
CCAC and CCDC for plurality
elections over swoon-SP societies are
  $\np$-complete.
\end{theorem}
\begin{theorem}\label{t:plurality-manymavericks}
For each $\epsilon > 0$,
CCAC and CCDC for plurality
elections over $I^\epsilon$-maverick-SP societies
are $\np$-complete, where $I$ denotes the input size.
\end{theorem}
For the Dodgson and PerceptionFlip notions of nearness to single-peakedness,
we can prove that allowing 
a constant number of 
adjacent-swaps (for each voter, separately, in the appropriate 
structure) still leaves the CCAC and CCDC problems in $\p$.
\tarkhide{(We mention in passing that the following result holds 
not just for constructive control, but even for 
the concept, which we are not focusing on in this paper,
of destructive control, 
in which one's goal is just to
preclude a certain candidate from winning.)}
\begin{theorem}\label{t:dp}
For each $k$,
CCAC and CCDC for plurality
elections are in $\p$ for 
Dodgson$_k$-SP societies and for 
PerceptionFlip$_k$-SP societies.
\end{theorem}

Our algorithm exploits the local nature of adding/deleting
candidates in single-peaked plurality elections.
Formally, \citet{fal-hem-hem-rot:j:single-peaked-preferences} 
made the following observation.

\begin{lemma}[Lemma~3.4 of~\citep{fal-hem-hem-rot:j:single-peaked-preferences}]
\label{lem:fhhr}
  Let $(C,V)$ be an election where $C = \{c_1, \ldots, c_m\}$ is a set
  of candidates, $V$
  is a collection of voters whose preferences are single-peaked with respect
  to societal axis $L$, and where $c_1\, L\, c_2\, L \,\cdots\, L\,
  c_m$. Within plurality, if $m \geq 2$ then
  \begin{enumerate}
  \item $\score_{(C,V)}(c_1) = \score_{(\{c_1,c_2\},V)}(c_1)$,
  \item for each $i$, $2 \leq i \leq m-1$, $\score_{(C,V)}(c_i) =
    \score_{(\{c_{i-1},c_{i},c_{i+1}\},V)}(c_i)$, and
  \item $\score_{(C,V)}(c_m) = \score_{(\{c_{m-1},c_m\},V)}(c_m)$.
  \end{enumerate}
\end{lemma}

It turns out that this local structure, slightly distorted, still
occurs in Dodgson$_k$-SP societies and in PerceptionFlip$_k$-SP
societies. Thus, our strategy for proving Theorem~\ref{t:dp} is to
first formally define what me mean by a distorted variant of the above
lemma, and then to adapt the CCAC and CCDC algorithms of
\citet{fal-hem-hem-rot:j:single-peaked-preferences} for single-peaked
societies to the distorted setting.

\begin{definition}
  Let $C = \{c_1, \ldots, c_m\}$ be a set of candidates and let
  $\sord$ be a linear order over $C$ (the societal axis) such that
  $c_1 \sord c_2 \sord \cdots \sord c_m$. For each $c_i \in C$ and
  each nonnegative integer $k$, $0 \leq k \leq m-1$, we define
  $N(\sord,C,c_i,k) = \{ c_j \mid |i-j| \leq k\}$. We call $N(\sord,
  C,c_i,k)$ the $k$-radius neighborhood of $c_i$ with respect to $C$ and
  $\sord$.
\end{definition}

\begin{definition}
  Let $E = (C,V)$ be a plurality election, 
  let $\sord$ be a linear order over $C$ (the societal axis),
  and let $k$ be a positive integer. We say that $E$ is $k$-local with
  respect to $L$ if for each $c \in C$ and each $C' \subseteq C$
  such that $c \in C'$ it holds that
  $\score_{(N(\sord,C',c,k),V)}(c) = \score_{(C',V)}(c).$
\end{definition}

In particular, the proof of Lemma~\ref{lem:fhhr}
(given in~\citep{fal-hem-hem-rot:j:single-peaked-preferences}) shows that
single-peaked plurality elections are $1$-local with respect to the societal
axis. We extend this result to Dodgson$_k$-SP societies and to
PerceptionFlip$_k$-SP societies.

\begin{lemma}\label{l:k-local-dodgson-perception}
  Let $k$ be a positive integer, let $E = (C,V)$ be a plurality election, and
  let $\sord$ be some linear order over $C$ (the societal axis). 
\begin{enumerate}
\item[(a)] If $E$ is Dodgson$_k$-SP with respect to $\sord$ then 
  $E$ is $(k+1)$-local with respect to $\sord$.
\item[(b)] If $E$ is PerceptionFlip$_k$-SP with respect to $\sord$
  then $E$ is $(k+1)$-local with respect to $\sord$.
\end{enumerate}
\end{lemma}
\begin{proof}
  Cases (a) and (b) are similar but not identical and thus we will
  handle each of them separately.

\smallskip\noindent\textbf{Case (a)}\quad Let $E = (C,V)$ be a Dodgson$_k$-SP election, where
    $C = \{c_1, \ldots, c_m\}$ and where the societal axis $\sord$ is
    such that $c_1 \sord c_2 \sord \cdots \sord c_m$.
Since for every $C' \subseteq C$, $(C',V)$ is Dodgson$_k$-SP with respect
to $\sord$ and since $E$ was chosen arbitrarily, it 
    suffices to show that for each candidate $c_i \in C$ it holds that
    $\score_{(N(\sord, C, c_i,k+1),V)}(c_i) = \score_{(C,V)}(c_i)$.

    Fix a candidate $c_i$ in $C$ and a voter $v$ in $V$. By definition
    of Dodgson$_k$-SP elections, there exists a vote $v'$ that is
    single-peaked with respect to $L$, such that $v$ can be obtained
    from $v'$ by at most $k$ swaps. Let $C' = N(\sord, C, c_i,k+1)$.
    We will show that $\score_{(C',\{v\})}(c_i) =
    \score_{(C,\{v\})}(c_i)$.

    First, if $v$ ranks $c_i$ on top among all candidates in $C$, then
    certainly $v$ ranks $c_i$ on top among candidates in $C'$. Thus,
    $\score_{(C',\{v\})}(c_i) \geq \score_{(C,\{v\})}(c_i)$. It remains to
    show that if $v$ does not rank $c_i$ on top among candidates in
    $C$, then $v$ also does not rank $c_i$ on top among candidates in
    $C'$. We consider two cases: 
    \begin{enumerate}
    \item The peak of $v'$ is not in $C'$. Then it is easy to verify that
      some candidate from $C'$ precedes $c_i$ in $v$. Otherwise, to
      convert $v'$ to $v$ one would need enough swaps for $c_i$ to
      pass $k+1$ candidates from $C'$ (either those ``to the left'' of
      $c_i$ in $\sord$ if the peak of $v'$ was ``to the left'' of
      $c_i$, or those ``to the right'' of $c_i$ if the peak of $v'$ was
      ``to the right'' of $c_i$).
    \item The peak of $v'$ is in $C'$. If $v$'s top-ranked candidate is in
      $C'$ then clearly $c_i$ is not ranked first among $C'$ in $v$.
      Thus, let us assume that the top ranked candidate in $v$ is not
      in $C'$ and that it is some candidate $c_j$. Without loss of
      generality, let us assume that $j > i+k+1$ (the case where $j <
      i-k-1$ is analogous). Let us also assume that the peak of $v'$
      is some candidate $c_{i+s}$, such that $1 \leq s \leq k+1$ (the
      case when $-k-1 \leq s \leq 0$ is impossible because converting
      $v'$ to $v$ requires at most $k$ swaps). The minimal number of
      swaps that convert $v'$ to a vote where $c_j$ is ranked first
      is at least $(k+1)-s+1 = k-s+2$ (these swaps involve $c_j$ and
      candidates $c_{i+s}, c_{i+s+1}, \ldots, c_{i+k+1}$. The minimum
      number of swaps in $v'$ that ensure that $c_i$ is ranked ahead
      of $c_{i+s}$ is at least $s$ (these swaps involve $c_i$ and
      candidates $c_{i+1}, c_{i+2}, \ldots, c_{i+s}$). Thus, the
      minimum number of swaps of candidates in $v'$ that ensure that
      $c_j$ is ranked first and that $c_i$ is ranked ahead of $c_{i+s}$ is
      $k+2$, which is more than the allowed $k$ swaps. Thus, this
      situation is impossible. As a result, some candidate from $C'$
      is ranked ahead of $c_i$ in $v$.
    \end{enumerate}
    Thus, we have shown that if $c_i$ is not ranked first in $v$ among
    the candidates from $C$, then $c_i$ is not ranked first in $v$
    among the candidates from $C'$. Thus, $\score_{(C',\{v\})}(c_i) \leq
    \score_{(C,\{v\})}(c_i)$, and with the previously shown inequality
    $\score_{(C',\{v\})}(c_i) \geq \score_{(C,\{v\})}(c_i)$, it must be
    the case that $\score_{(C',\{v\})}(c_i) =
    \score_{(C,\{v\})}(c_i)$. Since $v$ was chosen arbitrarily, we have
    that $\score_{(C',V)}(c_i) = \score_{(C,V)}(c_i)$. This completes the proof
    of part (a) of the theorem.

    \smallskip\noindent\textbf{Case (b)}\quad Let $E = (C,V)$ be an
    election, where $C = \{c_1, \ldots, c_m\}$.  Let us assume,
    without loss of generality, that $E$ is PerceptionFlip$_k$-SP via
    societal axis $\sord$, where $c_1 \sord c_2 \sord \cdots \sord
    c_m$.
Since for every $C' \subseteq C$, $(C',V)$ is PerceptionFlip$_k$-SP with respect
to $\sord$ and since $E$ was chosen arbitrarily, it 
suffices to show that for each candidate
$c_i \in C$ it holds that $\score_{(N(\sord, C, c_i,k+1),V)} =
\score_{(C,V)}(c_i)$.

    Let us fix a candidate $c_i$ in $C$ and a voter $v$ in $V$. This
    voter's preference order is single-peaked with respect to some
    order $\sord'$ that can be obtained from $\sord$ by at most $k$
    swaps of adjacent candidates. Let us assume that $c_{j'}$ is the
    candidate directly preceding $c_i$ in $\sord'$ and $c_{j''}$ is
    the candidate directly succeeding $c_i$ in $\sord'$ (in this proof
    we omit the easy-to-handle cases where $c_i$ is either the maximum
    or the minimum element of $\sord'$).

    We claim that for any $C' \subseteq C$ that includes $c_{j'}$,
    $c_{i}$, and $c_{j''}$, it holds that $\score_{(C',\{v\})}(c_i) =
    \score_{(C,\{v\})}(c_i)$. This is so because any voter that ranks
    $c_i$ on top, ranks $c_i$ on top irrespective of which other
    candidates are included. So, $\score_{(C',\{v\})}(c_i) \geq
    \score_{(C,\{v\})}(c_i)$. On the other hand, by Lemma~3.4 
    of~\citep{fal-hem-hem-rot:j:single-peaked-preferences},
    $\score_{(\{c_{j'},c_i,c_{j''}\},\{v\})}(c_i) =
    \score_{(C,\{v\})}(c_i)$. Thus, any voter that does not rank $c_i$
    on top, given a choice between $c_i$, $c_{j'}$, and $c_{j''}$
    ranks one of $c_{j'}$, $c_{j''}$ on top.  It is easy to see that
    $\{c_i, c_{j'}, c_{j''}\} \subseteq N(\sord, C, c_i, k+1)$, and
    so, $\score_{(C',\{v\})}(c_i) \leq \score_{(C,\{v\})}(c_i)$. 
    Thus, $\score_{(C',\{v\})}(c_i) = \score_{(C,\{v\})}(c_i)$ and since
    we picked $v$ arbitrarily, $\score_{(C',V)}(c_i) \leq
    \score_{(C,V)}(c_i)$. The proof of case (b) is complete.
\end{proof}

Now, Theorem~\ref{t:dp} is a consequence of the following, more
general, result.

\begin{theorem}\label{t:klocal}
For each constant $k$,
CCAC and CCDC for plurality
elections are in $\p$ for $k$-local elections. 
\end{theorem}

However, we do have $\np$-completeness if 
in Dodgson$_k$-SP societies or PerceptionFlip$_k$-SP societies
we allow the parameter $k$ to increase to $m - 2$, where
$m$ is the total number of candidates involved in the election.
(This many swaps allow us to, in effect, use the
same technique as for swoon-SP societies.)
\begin{theorem}\label{t:dodgson-C-2}
CCAC and CCDC for plurality
elections are $\np$-complete for 
Dodgson$_{m-2}$-SP societies and for 
PerceptionFlip$_{m-2}$-SP societies, where
$m$ is the total number of candidates involved
in the election.
\end{theorem}
Finally, we note that 
for single-caved elections,
CCAC and CCDC are in $\p$ for plurality.

\begin{theorem}\label{t:plurality-caved}
CCAC and CCDC for plurality
elections are in $\p$ for 
single-caved societies.
\end{theorem}

\def\bodyoffivepointone{{Bribery, negative-bribery, and strongnegative-bribery
for approval elections over log-maverick-SP societies are each in $\p$,
in both the standard and the marked model.}}

\section{Bribery}\label{sec:bribery}

We now briefly look at bribery of nearly single-peaked electorates,
focusing on
approval elections.  
For all three cases---bribery, negative-bribery, and
strongnegative-bribery---in which general-case 
$\np$-hardness bribery results have been shown to be in $\p$ for 
single-peaked societies~\citep{bra-bri-hem-hem:c:sp2},
we show that 
the complexity remains in polynomial time 
even if the number of mavericks is logarithmic in the
input size.
\begin{theorem}\label{t:approval-bribery}
\bodyoffivepointone
\end{theorem}

We mention in passing that although plurality bribery has never been
discussed with respect to single-peaked (or nearly single-peaked)
electorates, it is not hard to see that the two $\np$-complete bribery
cases for plurality (plurality-weighted-\$bribery and
plurality-weighted-negative-bribery) remain $\np$-complete on
single-peaked electorates, in one case immediately from a theorem of,
and in one case as a corollary to a proof of,
\citet{fal-hem-hem:j:bribery}.

\section{Related Work}\label{sec:related}

Although it has roots going even further back, the study of 
the computational complexity
of control and manipulation actions was started by a series of papers of
Bartholdi, Orlin, Tovey, and Trick around 1990 
\citep{bar-tov-tri:j:manipulating,bar-oli:j:polsci:strategic-voting,bar-tov-tri:j:control}
and the complexity of bribery was first studied far more 
recently~\citep{fal-hem-hem:j:bribery}.
For further references, history, context, and 
results regarding all of these, see the 
surveys~\citep{fal-hem-hem-rot:b-too-short:richer,fal-hem-hem:j:cacm-survey}.  
For example, it is known that there exist election systems 
that are resistant to many 
control attacks~\citep{erd-now-rot:j:sp-av,erd-rot:c:fallback,erd-pir-rot:t:bucklin,fal-hem-hem-rot:j:llull,hem-hem-rot:j:hybrid}.

The four papers most related to the present one are the following.
\citet{wal:c:uncertainty-in-preference-elicitation-aggregation}
insightfully raised the idea that general complexity results may
change in single-peaked societies.  His manipulative-action example
(STV) actually provides a case where single-peakedness fails to lower
manipulation complexity, but in a different context he did find a
lowering of complexity for single-peakedness.  The 
papers~\citep{fal-hem-hem-rot:cOUTDATEDbyJOURNALwithPTR:single-peaked-preferences,bra-bri-hem-hem:c:sp2}%
\tarkhideadd{ }{, the first of which was in the TARK conference, }
then broadly explored the effect of single-peakedness on manipulative
actions.  These three papers are all in the model of (perfect)
single-peakedness.  \citet{con:j:eliciting-singlepeaked},
in the context of preference elicitation, raised and experimentally
studied the issue of \emph{nearly} single-peaked societies.
\citet{esc-lan-ozt:c:single-peaked-consistency} 
also discussed nearness
to single-peakedness, and the papers
\cite{fal-hem-hem-rot:j:single-peaked-preferences,bra-bri-hem-hem:c:sp2}
both raise as open issues 
whether 
shield-evaporation (complexity) 
results for single-peakedness
will withstand near-single-peakedness.
The
present paper seeks to bring the ``nearly single-peaked'' lens to the
study of manipulative actions.

It is well-known (see \citet{elk-fal-sli:c:votewise-dr} 
and the references therein) that many election systems are defined, or
can be equivalently defined, as selecting whichever candidate's region
of being a winner under some notion of ``consensus'' has a vote
set that is ``closest'' to the input vote set, where ``closest'' is
defined by applying some norm (e.g., sum or max) to a vector whose
$i$th component is some notion of distance (e.g., number of
adjacent-swaps 
to get between the votes) between the $i$th votes
on each list.  We note that, similarly, most of our notions of nearness 
to single-peaked can be framed as saying that the input vote set is close
(in the same sense) to some vote set that is consistent with the input
societal linear order.  The parallel isn't perfect, since in the
former work there are multiple target regions and the minimum over
them is crucially selecting the winner; 
and also, approaches focusing on distance typically require
commutativity of the distance function, but notions of diverging from
a societal order may, as a matter of human behavior, be asymmetric.
(The swoon notion is asymmetric; often one can 
swoon from $v_1$ to $v_2$, but not vice versa).  
But we mention
that most of our norms/distances have already proven useful in other
election contexts, and that if one finds additional norms/distances 
natural here, then one could study what happens under 
those.

This paper is focused on the line of work that looks at whether
single-peakedness removes $\np$-hardness results about manipulative
actions.  \emph{Most of our key results are cases where we show that even
\emph{nearly}-single-peaked elections fall to $\p$.  Our results of
that sort are polynomial-time upper bounds, and so apply on 
\emph{all} inputs.}
However, a few of our results are about $\np$-hardness.  For those, it is
important to mention that $\np$-hardness is a worst-case theory, and so
such results are just a first step on a path that one hopes may
eventually reach some notion of average-case hardness.  That is a
long-term and difficult goal, but has not yet been proven impossible.
Although many people have the impression that recent results such as
\citet{fri-kal-nis:c:quantiative-gib-sat} prove that any reasonable
election system can often be manipulated, that powerful paper, for example,
merely proves that the manipulation probability 
\emph{cannot go 
very quickly
to zero} asymptotically---it does not prove that the manipulation
probability asymptotically 
does not go to zero.
Other work that controls the candidates-to-voters cardinality 
relation
has experimentally suggested
stronger claims in certain settings, 
but 
is of necessity within the setting of making assumptions about 
the distribution of votes
(see, e.g., \citet{wal:c:where-hard-veto}), and typically
presents simulations but not theorems and proofs.

Finally, given the high hopes of many people for 
heuristic algorithms, it is extremely important to 
understand that it follows from known complexity-theoretic results
that (unless shocking complexity-class collapses occur)
\emph{no polynomial-time algorithm can come too
close to accepting any NP-hard set}.
In particular, if any polynomial-time heuristic 
correctly solves any $\np$-hard problem on all but a 
sparse (i.e., at most a polynomial number of strings of each length)
set of inputs,
then $\p = \np$
(\citet{sch:j:closeness} established a precursor of this
result that held for appropriately ``paddable'' NP-hard sets,
and the stronger claim we have stated here
follows from the 1-truth-table special case of the 
result of \citet{ogi-wat:j:pbt}
that if any NP-bounded-truth-table-hard set is sparse 
then $\p=\np$).
And so such extremely good heuristic algorithms almost certainly do
not exist.  But what about less ambitious hopes?  Can any
polynomial-time algorithm agree with any NP-hard set except for at
most $n^{\log^{\bigo(1)}n}$ strings 
at each length, i.e., the
symmetric difference between the set accepted by the algorithm and the
NP-hard set has density $n^{\log^{\bigo(1)}n}$
(equivalently, $2^{\log^{\bigo(1)}n}$)?  The answer is ``no,''
unless shocking complexity-theoretic collapses occur.  To see this,
we simply need to note that if such an algorithm existed, then the
NP-hard set would  polynomial-time 1-truth-table reduce
(indeed, it would even reduce by a
polynomial-time 
1-truth-table reduction that was in addition restricted to a single
truth-table, namely, the parity truth-table) to a set of density
$n^{\log^{\bigo(1)}n}$.  However, it is known
(\cite{hom-lon:j:sparse,buh-hom:c-short:superpolynomial}) that if any
NP-hard set even polynomial-time 
$\bigo(1)$-truth-table reduces (i.e., polynomial-time bounded-truth-table
reduces) to a set of density $n^{\log^{\bigo(1)}n}$, then 
(i)~all NP sets can be deterministically solved in time
$n^{\log^{\bigo(1)}n}$, and 
(ii)~$\rm EXP = NEXP$
(where ${\rm EXP} = \bigcup_{\textrm{\,polynomials $p$}}
{\rm DTIME}[2^{p(n)}]$ and 
${\rm NEXP} = \bigcup_{\textrm{\,polynomials $p$}}
{\rm NTIME}[2^{p(n)}]$),
i.e., deterministic and nondeterministic exponential time coincide.
Both of these consequences are broadly believed not to hold.
(Various additional unlikely collapse
consequences follow for the case
where 
we are speaking not merely of NP-hard sets but in fact of NP-complete 
sets~\cite{buh-hom:c-short:superpolynomial}, such as the few 
hard problems discussed in this paper.)
So polynomial-time heuristic algorithms are extremely
unlikely to be able to come within $n^{\log^{\bigo(1)}n}$
errors-per-length of any NP-hard set.\footnote{One may 
wonder how close 
to NP-hard sets \emph{can} polynomial-time heuristic algorithms 
easily come?  By using 
easy ``padding'' constructions, it is not hard to see that 
for each positive integer $k$ it holds that 
every 
NP-hard set is polynomial-time equivalent to 
an NP-hard set that differs from the empty set---which itself 
certainly is in P---on at most $2^{n^{1/k}}$ of the $2^n$ 
length-$n$ strings, i.e., they agree on all but 
an exponentially vanishing portion of the domain; 
see the construction in Footnote 10 
of \cite{erd-hem-rot-spa:j:junta} 
for how to do this, 
except with the ``2'' there being changed to $k+1$
(see also the comments/discussion in Appendix~C
of~\cite{erd-hem-rot-spa:t2:lobbying-but-now-is-avg-and-juntas}).
The
best
way to interpret this is not to just say that all NP-hard sets
are akin to easy sets, but rather to realize that frequency 
of easiness is not a robust concept with respect to 
polynomial-time reductions, and neither is average-time 
complexity computed simply by averaging.  Indeed,
this type of effect is why average-case complexity
theory is defined in ways far more subtle and complex than 
just taking a straightforward averaging of 
running times~\cite{lev:j:average-case}.}

\section{Conclusions and Open Directions}\label{sec:conclusions-and-open}

Motivated by the fact that real-world electorates are unlikely to be
flawlessly single-peaked, we have studied the complexity of
manipulative actions on nearly single-peaked electorates.
We observed a wide range of behavior.
Often, a modest amount of non-single-peaked behavior is not enough to
obliterate an existing polynomial-time claim.  We find this the most
important theme of this paper---its ``take-home message.''  
So if one feels that previous
polynomial-time manipulative-action algorithms for single-peaked
electorates are suspect since real-world electorates tend not to be
truly single-peaked but rather nearly single-peaked, our results of
this sort should reassure one on this point---although they are but a
first step, as the paragraph after this one will explain.  Yet we 
also found that 
sometimes allowing even one deviant voter is enough to raise the
complexity from $\p$ to $\np$-hardness, and sometimes allowing any number of
deviant voters has no effect at all on the complexity.  
We also saw cases where frequency of mavericity extracted a 
price in terms of amount of nondeterminism used.  We feel this is 
a connection that should be further explored, and regarding 
Corollary~\ref{c:stcp-beta},
we particularly 
commend to the reader's attention the issue of proving 
completeness for---not merely membership in---the levels of the 
beta hierarchy.  We conjecture that completeness holds.

One might wish to study other
notions of closeness to single-peakedness and, in particular, 
one might want to
combine our notions.  Indeed,
in real human elections, there probably are both mavericks and
swooners, and so our models are but a first step.  In addition, the
types of nearness that appear in different human contexts may differ
from each other, and from the types of nearness that appear in
computer multiagent systems contexts.  Models of
human/multiagent-system behavior, and empirical study of actual
occurring vote sets, may help identify the most appropriate notions of
nearness for a given setting.

Our control work studies just one type of control-attack at a time.
We suspect that many of our polynomial-time results could be 
extended to handle multiple types of attacks simultaneously,
as has recently been explored (without single-peakedness constraints)
by 
\citet{fal-hem-hem:c:multimode}, and we mentioned in passing
one result for which we have already shown this.

\acknowledgements

\pagebreak[3]

\appendix

\section{Formal Definitions}\label{sec:formal}

In this section we provide the missing formal definitions of the
problems that we study (variants of manipulation, control, and
bribery) and formal definitions of the problems we reduce from in our
hardness proofs.

\begin{definition}[\citep{con-lan-san:j:when-hard-to-manipulate}]
  Let $R$ be an election system. 
  In the CCWM problem for $R$ we are given a set of
  candidates $C$, a preferred candidate $p \in C$, a collection of
  nonmanipulative voters $S$ (each vote consists of a preference order
  and a positive integer, the weight of the vote), and a collection $T$ of $n$
  manipulators, each specified by its positive integer weight.
  We ask if it is possible
  to set the preference orders of the manipulators in such a way that
  $p$ is a winner of the resulting $R$ election $(C, S \cup T)$.
\end{definition}

The following control notions are due to 
the seminal paper of \citet{bar-tov-tri:j:control}, except
the notion below of CCAC follows \citet{fal-hem-hem-rot:j:llull} in
employing a bound, $K$, to make it better match the other control
types.

\begin{definition}[\citep{bar-tov-tri:j:control}]
  Let $R$ be an election system. 
  \begin{enumerate}
  \item[(a)] In the CCAC problem for $R$ we are given two disjoint
    sets of candidates, $C$ and $A$, a collection $V$ of votes over $C \cup
    A$, a candidate $p \in C$, and a nonnegative integer $K$. We ask
    if there is a set $A' \subseteq A$ such that (a) $\|A'\| \leq K$,
    and (b) $p$ is a winner of $R$ election $(C \cup A', V)$.

  \item[(b)] In the CCDC problem for $R$ we are given an election
    $(C,V)$, a candidate $p \in C$, and a nonnegative integer $K$. We
    ask if there is a set $C' \subseteq C$ such that (a) $\|C'\| \leq
    K$, (b) $p \notin C'$, and (c) $p$ is a winner of $R$ election $(C
    - C', V)$.

  \item[(c)] In the CCAV problem for $R$ we are given a set of
    candidates $C$, two collections of voters, $V$ and $W$, over $C$, a
    candidate $p \in C$, and a nonnegative integer $K$. We ask if
    there is a subcollection $W' \subseteq W$ such that (a) $\|W'\| \leq K$, and
    (b) $p$ is a winner of $R$ election $(C, V\cup W')$.

  \item[(d)] In the CCDV problem for $R$ we are given an election
    $(C,V)$, a candidate $p \in C$, and a nonnegative integer $K$. We
    ask if there is a collection $V'$ of voters that can be obtained from
    $V$ be deleting at most $K$ voters such that $p$ is a winner of
    $R$ election $(C, V')$.
  \end{enumerate}
\end{definition}

The bribery notions below are due to 
\citet{fal-hem-hem:j:bribery},
except
the notion below of negative and strongnegative bribery for 
approval voting are due to 
\citet{bra-bri-hem-hem:c:sp2}.

\begin{definition}[\citep{fal-hem-hem:j:bribery,bra-bri-hem-hem:c:sp2}]
  Let $R$ be an election system.  In the weighted-\$bribery problem
  for $R$ we are given an election $(C,V)$, where each vote consists
  of the voter's preferences (as appropriate for the election system,
  e.g., an approval vector for approval voting and a preference order
  for plurality) and two integers (this vote's positive integer weight and
  this vote's nonnegative integer price), a candidate $p \in C$, and a
  nonnegative integer $K$ (the allowed budget). We ask if there is a
  subcollection of votes, whose total price does not exceed $K$, such
  that it is possible to ensure that $p$ is an $R$-winner of the
  election by modifying the preferences of these votes.

  The problems (a) weighted-bribery, (b) \$bribery, and (c) bribery for
  $R$ are variants of weighted-\$bribery for $R$ where, respectively:
  (a) no prices are specified and each vote is treated
  as having unit cost, (b) no weights are specified, and each vote
  is treated as having unit weight, 
  and (c) no prices or weights are specified, and each vote
  is treated as having unit price and unit weight.

  For plurality, ``negative'' bribery means no bribed voter can
  have $p$ as the most preferred candidate in his/her preference order.

  For approval voting, ``negative'' bribery
  means a bribe cannot change someone from disapproving of $p$ to
  approving of $p$, and ``strongnegative'' bribery means
  every bribed voter must
  end up disapproving of $p$.  

\end{definition}

\begin{definition}[see, for example, \citet{gar-joh:b:int}]
  A PARTITION instance $I$ is a set of $\{k_1, \ldots, k_n\}$
  of $n$ distinct positive integers that sums to $2K$.
  $I$ is a yes instance
  if there exists a subset of $\{k_1, \ldots, k_n\}$ that 
  sums to $K$ and it is a no instance otherwise.
  An X3C instance $I = (B, \calS)$ consists of a base set $B =
  \{b_1, \ldots, b_{3k}\}$ and a family $\calS = \{S_1, \ldots,
  S_n\}$ of three-element subsets of $B$. $I$ is a yes instance
  if it is possible to pick exactly $k$ sets from $\calS$ so that
  their union is $B$ and it is a no instance otherwise.
\end{definition}

\begin{definition}
For score-based election systems (e.g., plurality, approval, scoring
protocols), we write $\score_{(C,V)}(c)$ to denote the score of
candidate $c$ in election $(C,V)$; naturally we require that $c \in
C$. The particular election system that we use will always be clear
from context.
\end{definition}

\section{Proofs from Section~\ref{sec:manip}}

\newtheorem*{theorem3x1}{Theorem~\ref{t:1msp-dich}}
\newtheorem*{theorem3x2}{Theorem~\ref{t:veto-easy}}
\newtheorem*{theorem3x3}{Theorem~\ref{t:veto-hard}}
\newtheorem*{corollary3x4}{Corollary~\ref{t:veto-cor}}
\newtheorem*{theorem3x5}{Theorem~\ref{t:swoon-ccwm-veto}}
\newtheorem*{observation3x6}{Observation~\ref{o:three-cand}}
\newtheorem*{theorem3x7}{Theorem~\ref{t:dodgson-ccwm-veto}}
\newtheorem*{theorem3x8}{Theorem~\ref{t:manip-sc}}

\begin{theorem3x1}
  For each $\alpha_1 \geq \alpha_2 > \alpha_3$, $\ccwm$ for
  $(\alpha_1,\alpha_2,\alpha_3)$ elections over 1-maverick-SP societies is
  $\np$-complete.
\end{theorem3x1}

\begin{proof}[Proof of Theorem~\ref{t:1msp-dich}]
Without loss of generality, we assume that $\alpha_3 = 0$.
We will reduce from PARTITION.  Given a set
$\{k_1, \ldots, k_n\}$ of $n$ distinct positive integers that sums to $2K$,
define the following instance of CCWM.  Let $C = \{p, a, b\}$,
let society's order be $a L p L b$, let $S$ consist of one voter
with preference order $a > b > p$ (note that this voter is the maverick)
with weight $(2\alpha_1 - \alpha_2)\alpha_1 K$, and one
voter with preference order $b > p > a$
with weight $(2\alpha_1 - \alpha_2)(\alpha_1 - \alpha_2) K$.
(Technically, weights need to be positive, but if $\alpha_1 = \alpha_2$
we can get the same effect by letting $S$ consist of just the maverick.)
Note that $\score_{(C,S)}(a) = \score_{(C,S)}(b) =
(2 \alpha_1^3 - \alpha_1^2 \alpha_2) K$
and that
$\score_{(C,S)}(p) = (2 \alpha_1 - \alpha_2) (\alpha_1 - \alpha_2) \alpha_2 K =
(2 \alpha_1^2 \alpha_2 - 3 \alpha_1 \alpha_2^2  + \alpha_2^3) K$.
Let $T$ consist of $n$ manipulators with weights
$(\alpha_1^2 - \alpha_1\alpha_2 + \alpha_2^2) k_1, \ldots,
(\alpha_1^2 - \alpha_1\alpha_2 + \alpha_2^2) k_n$.

If there is a subset of $k_1, \ldots, k_n$ that sums to $K$,
then we let all manipulators in $T$ whose weight divided by
$(\alpha_1^2 - \alpha_1\alpha_2 + \alpha_2^2)$
is in this
subset vote $p > a > b$, and all manipulators in $T$ whose
weight divided by
$(\alpha_1^2 - \alpha_1\alpha_2 + \alpha_2^2)$
is not in this subset vote $p > b > a$.  
It is immediate that 
$\score_{(C,S \cup T)}(a) = \score_{(C,S \cup T)}(b) =
\score_{(C,S)}(a) + (\alpha_1^2\alpha_2 - \alpha_1\alpha_2^2 + \alpha_2^3) K =
(2 \alpha_1^3 - \alpha_1 \alpha_2^2 + \alpha_2^3) K$
and that $\score_{(C,S \cup T)}(p) = \score_{(C,S)}(p) + 
(2\alpha_1^3 - 2\alpha_1^2\alpha_2 + 2 \alpha_1 \alpha_2^2) K =
(2 \alpha_1^3 - \alpha_1 \alpha_2^2 + \alpha_2^3) K$.  It follows
that all candidates are tied, and thus all candidates are winners.

For the converse, suppose the manipulators vote so that $p$ becomes
a winner.  It is easy to see that we can assume that all manipulators
rank $p$ first.  From the calculations above, it is 
also easy to see that it is always the case that $2\score_{(C,S \cup T)}(p)
\leq \score_{(C,S \cup T)}(a) + \score_{(C,S \cup T)}(b)$.  In order for $p$
to be a winner, we thus certainly need the scores of $a$ and $b$ to be equal.
This implies that $\score_{(C,T)}(a) =  \score_{(C,T)}(b)$.  But then the
weights of the manipulators voting $p > a > b$ sum to
$(\alpha_1^2 - \alpha_1\alpha_2 + \alpha_2^2) K$.
\end{proof}

\begin{theorem3x2}
For each $k \geq 0$ and $m \geq k + 3$, $\ccwm$ for
$m$-candidate veto elections over $k$-maverick-SP 
societies is in $\p$.
\end{theorem3x2}

\begin{proof}
Let $k \geq 0$ and $m \geq k+3$, and
let society's order be $L$.  Let $c_{\ell}$ be the leftmost candidate
in $L$ and let $c_{r}$ be the rightmost candidate in $L$. 
In an $m$-candidate veto election, 
the $m-2$ candidates in $C - \{c_{\ell},c_r\}$ are never vetoed
by the nonmavericks.  Every maverick vetoes at most one of
these $m-2$ candidates. Since $k < m - 2$, 
in a $k$-maverick $m$-candidate veto election, 
there exists a candidate that is never vetoed.
So, given an instance of $\ccwm$ for $m$-candidate veto elections
over $k$-maverick-SP societies, $p$ can be made a winner if and only if
$p$ is never vetoed by the nonmanipulators.
\end{proof}

\begin{theorem3x3}
For each $k \geq 0$ and $m \geq 3$ such that $m \leq k + 2$, $\ccwm$ for
$m$-candidate veto elections over $k$-maverick-SP societies 
is $\np$-complete.
\end{theorem3x3}

\begin{proof}
We will again reduce from PARTITION.  Given a set
$\{k_1, \ldots, k_n\}$ of $n$ distinct positive integers that sums to $2K$,
define the following instance of CCWM.  Let $C = \{p, a, b,
c_1, \ldots, c_{m-3}\}$,
let society's order be $a L p L c_1 L \cdots L c_{m-3} L b$,
let $S$ consist of $m-2$ voters each of weight $K$.
For every candidate $c$ in $C - \{a,b\}$ there is a voter in $S$
that ranks $c$ last.  Note that all voters in $S$ are mavericks.
This is allowed, since $m-2 \leq k$.
Let $T$ consist of $n$ manipulators with weights 
$k_1, \ldots, k_n$.

If there is a subset of $k_1, \ldots, k_n$ that sums to $K$,
then we let all manipulators in $T$ whose weight is in this
subset vote $a > p > c_1 > \cdots > c_{m-3} > b$,
and all manipulators in $T$ whose weight is not in this subset
vote  $b > c_{m-3} > \cdots > c_1 > p > a$.
It is immediate that all candidates in election $(C,S\cup T)$ are tied
and so $p$ is a winner.

For the converse, suppose the manipulators can vote so that 
$p$ becomes a winner.  Note that $p$ needs to gain at least $K$
points over $a$ and over $b$ in $T$. Clearly, the only way this can
happen is if $\score_{(C,T)}(a) = \score_{(C,T)}(b) = K$.
But then the weights of the voters in $T$ who rank $a$ last add
to~$K$.
\end{proof}

\begin{theorem3x5}
  For each $m \geq 5$, $\ccwm$ for
  $m$-candidate veto elections
  in swoon-SP societies is in $\p$\@.
  For $m \in \{3,4\}$, this problem is $\np$-complete.
\end{theorem3x5}

\begin{proof}
First suppose that $m \geq 5$.  Let $L$ be society's order.
Let $c$ be a candidate such that there are at least two candidates
to the left of $c$ in $L$ and there are at least two candidates
to the right of $c$ in $L$.
In an $m$-candidate veto election in a swoon-SP society, $c$ 
will never be vetoed.  Given a $\ccwm$ instance for 
$m$-candidate veto elections
in swoon-SP societies, $p$ can be made a winner if and only if
$p$ is never vetoed by the nonmanipulators.

Now consider the case that $m = 4$.  We will reduce from
PARTITION. Given a set $\{k_1, \ldots, k_n\}$ of $n$ distinct
positive integers
that sums to $2K$,
define the following instance of CCWM.  Let $C = \{p, a, b, c\}$,
let society's order be $a L p L b L c$, let $S$ consist of two voters,
each with weight $K$.  One voter in $S$ votes
$a > c > b > p$ and one voter votes $c > a > p > b$.
Let $T$ consist of $n$ manipulators with weights $k_1, \ldots, k_n$.

If there is a subset of $k_1, \ldots, k_n$ that sums to $K$,
then we let all manipulators in $T$ whose weight is in this
subset veto $a$ and
all manipulators in $T$ whose weight is not in this subset
veto $c$.  It is immediate that all candidates in election
$(C,S\cup T)$ are tied and so $p$ is a winner.

For the converse, suppose the manipulators can vote so that 
$p$ becomes a winner.  Note that $p$ needs to gain at least $K$
points over $a$ and over $c$ in $T$. Clearly, the only way this can
happen is if $\score_{(C,T)}(a) = \score_{(C,T)}(c) = K$.
But then the weights of the voters in $T$ who veto $a$ add
to~$K$.

A very similar proof can be used to show the statement for $m = 3$.
However, the statement for $m = 3$ also follows immediately from
the fact that $\ccwm$ for $3$-candidate veto elections
is $\np$-complete and the observation below that
every 3-candidate election is a swoon-SP election.
\end{proof}

\begin{observation3x6}
Every 3-candidate election is a swoon-SP election and
a Dodgson$_1$-SP election and so all complexity results for 3-candidate
elections in the general case also hold for 
swoon-SP elections and Dodgson$_1$-SP elections. 
\end{observation3x6}

\begin{proof}
Since every 2-candidate vote is single-peaked, it follows
immediately that every 3-candidate election is a swoon-SP election.
For the Dodgson$_1$-SP case, suppose society's order is $a L b L c$.
The only 
votes 
that are not single-peaked 
are $a > c > b$ and 
$c > a > b$.  But note that both of these
votes are one adjacent swap away from being single-peaked, by swapping
the last two candidates.
\end{proof}

\begin{theorem3x7}
  For each $m \geq 5$, $\ccwm$ for
  $m$-candidate veto elections
  in Dodgson$_1$-SP societies is in $\p$\@.
  For $m \in \{3,4\}$, this problem is $\np$-complete.
\end{theorem3x7}

\begin{proof}
The $m \geq 5$ case follows using the same proof as the 
$m \geq 5$ case of Theorem~\ref{t:swoon-ccwm-veto}.
The $m = 4$ case follows using the same proof as
the $m = 4$ case of Theorem~\ref{t:swoon-ccwm-veto}
except that, in order for the votes to be within one adjacent swap
of being consistent with societal order,
the two voters in $S$ now vote $c > b > a > p$ and $a > p > c > b$.
The $m = 3$ case follows from Observation~\ref{o:three-cand}.
\end{proof}

\begin{theorem3x8}
  For each $\alpha_1 \geq \alpha_2 > \alpha_3$, $\ccwm$ for
  $(\alpha_1,\alpha_2,\alpha_3)$ elections over single-caved societies is
  $\np$-complete if 
  $(\alpha_1 - \alpha_3) \leq 2(\alpha_2-\alpha_3)$  
  and otherwise is in $\p$.
\end{theorem3x8}

\begin{proof}
Without loss of generality, assume that $\alpha_3 = 0$.
We first consider the case that $\alpha_1 > 2 \alpha_2$.
Let $(C,V)$ be an 
$(\alpha_1,\alpha_2,\alpha_3)$ election over single-caved societies,
let $W$ be the total weight of $V$, 
and let $L$ be society's order.  Consider the middle candidate
in $L$.  This candidate can never be ranked first, and so its score
will be at most $\alpha_2 W$,
and the sum of the scores of the other two candidates will
be at least $\alpha_1 W$.  Since $\alpha_1 > 2 \alpha_2$, it follows
that the middle candidate will never be a winner if $W > 0$.
Given an instance of $\ccwm$ for
$(\alpha_1,\alpha_2,0)$ elections over single-caved societies,
$p$ can be made a winner if and only if
(1) $p$ is the middle candidate in $L$ and $W = 0$, or
(2) $p$ is not the middle candidate in $L$ and $p$ is a winner
if all manipulators rank $p$ first, then the middle candidate, and then
the last candidate.  All this is easy to check in polynomial time.

Now consider the case that $\alpha_1 \leq 2 \alpha_2$.
We will show that in this case PARTITION can be reduced to 
$\ccwm$ for $(\alpha_1,\alpha_2,0)$ elections over single-caved societies.
Given a set $\{k_1, \ldots, k_n\}$ of $n$ distinct positive integers
that sums to $2K$,
define the following instance of CCWM.  Let $C = \{p, a, b\}$ and
let society's order be $a L p L b$, and  let $S$ consist of two voters,
each with weight $(2\alpha_2 - \alpha_1)K$.  One voter in $S$ votes
$a > b > p$ and one voter votes $b > a > p$. (Technically, weights need
to be positive, but if $\alpha_1 = 2\alpha_2$ we can get the
same effect by letting $S = \emptyset$.)
Let $T$ consist of $n$ manipulators with weights $(\alpha_1 + \alpha_2)k_1,
\ldots,(\alpha_1 + \alpha_2)k_n$.

If there exists a subset of $\{k_1, \ldots, k_n\}$ that sums to $K$, we let 
the manipulators whose weight
divided by  $(\alpha_1 + \alpha_2)$
is in this subset vote $a > p > b$ and
the manipulators whose weight
divided by  $(\alpha_1 + \alpha_2)$
is not in this subset vote $b > p > a$.
It is easy to see that
$\score_{(C,S\cup T)}(a) = 
\score_{(C,S\cup T)}(b) = 
(2\alpha_2 - \alpha_1)(\alpha_1 + \alpha_2)K +
\alpha_1 (\alpha_1 + \alpha_2) K = 2 \alpha_2 (\alpha_1 + \alpha_2) K = 
\score_{(C,S\cup T)}(p)$, and so $p$ is a winner.

For the converse, suppose $p$ can be made a winner.
Since $p$ is the middle candidate in $L$, $p$ can not be ranked first.
Without loss of generality we can assume that the voters in $T$ vote
$a > p > b$ or $b > p > a$.  It follows that
$\score_{(C,S\cup T)}(p) = 2 \alpha_2 (\alpha_1 + \alpha_2) K$.
Since $\score_{(C,S\cup T)}(a) + \score_{(C,S\cup T)}(b) 
= 4 \alpha_2 (\alpha_1 + \alpha_2) K$, by the argument above,
the only way $p$ can be a winner is if $a$ and $b$ tie in $T$.
But then the weights of the manipulators voting $a > p > b$ sum
to $(\alpha_1 + \alpha_2)K$.
\end{proof}

\section{Proofs from Section~\ref{sec:control}}
\label{app:stcp}

In the following subsections we provide the missing proofs from
Section~\ref{sec:control}.

\newtheorem*{theorem4x1}{Theorem~\ref{t:stcp}}
\newtheorem*{theorem4x4}{Theorem~\ref{t:plurality-maverick}}
\newtheorem*{theorem4x4new}{Theorem~\ref{t:stcp-condorcet}}
\newtheorem*{theorem4x5}{Theorem~\ref{t:plurality-swoon}}
\newtheorem*{theorem4x6}{Theorem~\ref{t:plurality-manymavericks}}
\newtheorem*{theorem4x7}{Theorem~\ref{t:dp}}
\newtheorem*{theoremklocal}{Theorem~\ref{t:klocal}}
\newtheorem*{theorem4x8}{Theorem~\ref{t:dodgson-C-2}}
\newtheorem*{theorem4x9}{Theorem~\ref{t:plurality-caved}}

\subsection{Proofs of Theorems~\ref{t:stcp},~\ref{t:stcp-condorcet},
and~\ref{t:plurality-maverick}}

\begin{theorem4x1}
CCAV and CCDV for approval
elections over log-maverick-SP societies are each 
in $\p$.  For CCAV, the complexity remains in $\p$
even for the case where no limit is imposed on the 
number of mavericks in the initial voter set, and the
number of mavericks in the set of potential
additional voters is logarithmically 
bounded (in the overall problem input size).
\end{theorem4x1}

\begin{proof}
  Consider the case of CCAV\@.  We will handle directly the stronger case
  mentioned in the theorem, namely the one with no limit on the
  number of mavericks in the initial voter set.  There of course will
  be a 
$\bigo(\log(\mathrm{ProblemInputSize}))$
limit on the number of mavericks
  in the set of voters to potentially add.  Let that (easy, 
  nondecreasing) upper-bounding
  function be called $f$.  

  So, suppose we are given an input instance of this problem.  Let
  $K$, which is part of the input, be the limit on the number of
  voters we are allowed to add.  We start by doing the obvious
  syntactic checks, and we also check that the number of voters in the
  additional voter set who are not consistent with the input societal
  order is at most $f(\mathrm{ProblemInputSize})$.  If any of these 
  checks fail, we reject.  

  Now, we will show how to build a polynomially long list of instances
  of the CCAV problem over single-peaked elections such that our
  control goal is possible to achieve \emph{if and only if} one or
  more of those control problems has a goal that can be 
  achieved.  (That is, we will implicitly give a polynomial-time
  disjunctive truth-table reduction to the single-peaked case.)

  Our construction is as follows.  For each choice of which mavericks
  {}from the additional voter set to add to our election, we will 
  generate at most one instance of a single-peaked control question.
  Since there are at most logarithmically many such mavericks, 
  and the number of cases we have to look at is the cardinality 
  of the powerset of the number of mavericks among the additional
  voters, the number of instances we generate is polynomially bounded.

  For each choice $A$ of which mavericks from among the additional
  voter set to add to the main election, we generate at most one
  instance as follows.  If $\|A\| > K$, we will generate no instance,
  as that choice is trying to add illegally many additional voters.
  Otherwise, we generate a single-peaked election instance as follows.
  Move the elements of $A$ from the additional voter set to the main
  election.  Remove all remaining mavericks from the additional voter
  set.  Demaverickify our election as follows: For each maverick voter
  $v$, for each candidate $c$ that $v$ approves, add a new voter who
  approves of only $c$ (and so certainly is consistent with the
  single-peaked societal order).  Then remove all the maverick voters.
  Note that this demaverickification process does not change the
  approval counts of the election and does ensure that the electorate is
  single-peaked.  The entire demaverickification process does not
  increase the problem's size by more than a polynomial, since no
  voter is replaced by more than $\|C\|-1$ voters.
  Replace $K$ by $K - \|A\|$.  The resulting instance
  is the instance that this choice of $A$ adds to our collection of
  instances.

  So, we have created a polynomial-length list of (polynomial-sized) 
  instances of the single-peaked CCAV problem.  It is easy
  to see that our control goal can be achieved exactly if for at least
  one of these instances the control goal can be achieved.  Briefly
  put, that is because our problem has a successful control action
  (after passing the initial maverick-cardinality-limit check) exactly
  if there is some appropriate-sized subset of additional voters that
  we can add to make the favored candidate become a winner.  Our above
  process tries every legal set of choices for which mavericks from
  the additional voter set might be the mavericks in the added set.
  And the instance it generates, based on that choice, will have a
  successful control action precisely if what remains of our initial
  $K$ bound, after we remove the cardinality of the added mavericks,
  is such that there is some number of nonmaverick additional voters
  who can be added to achieve the desired victory for $p$.  In
  addition, the instance generated is a single-peaked society, and the
  transformation we used to make it single-peaked doesn't in anyway
  affect the answer to the created instance, since the
  demaverickification occurred only on voters that were (at that
  point, although some had not started there) in our main voter set,
  and the only affect that set has on the single-peaked CCAV control
  question is the approval totals of each candidate, and our
  demaverickification did not alter those totals.

  Our polynomial-length list of instances is composed just of
  instances of the CCAV problem for approval elections over
  single-peaked electorates.  That problem has a polynomial-time
  algorithm~\citep{fal-hem-hem-rot:j:single-peaked-preferences}.  And
  so we run that algorithm on each of the polynomially many instances,
  and if any finds a successful control action, our original problem
  has a successful control action, and if not our original problem
  does not.  Thus, our proof of the CCAV case is complete.

  However, a final comment is needed, since we wish to not only give a
  yes/no answer, but to in fact find what control action to take, when
  one is possible.  (Doing so goes beyond what the theorem promises,
  but we in general will try to give algorithms that not only 
  give yes/no answers but also that at least implicitly make available 
  the actual successful actions for the yes instances.)
  Formally speaking, disjunctive truth-table
  reductions are about languages, rather than about solutions.
  Nonetheless, from our construction it is immediately clear how a
  successful control action for any problem on the list---and the
  polynomial-time algorithm
  of \citet{fal-hem-hem-rot:j:single-peaked-preferences} in fact gives
  not merely a yes/no answer but in fact finds a successful control
  action when one exists---specifies a successful control action for
  our nearly-single-peaked original problem.

  The CCDV case might at first seem to be almost completely analogous,
  except that in that case, there is no separate pool of additional
  voters, and the logarithmic bound applies to the entire set of
  initial voters.  However, the reduction approach we took above for
  CCAV at first seems not to work here.  The reason is that for the
  CCAV case, the mavericks we added could be demaverickified in a way
  that didn't interfere with the call to the single-peaked case of the
  CCAV approval voting algorithm, \emph{and the mavericks we decided not to add
  could (for the instance being generated) be deleted}.  In contrast,
  for the CCDV case, whichever mavericks we don't delete remain very
  much a part of the election---and are indeed part of the
  instance we would like to generate of a case of a call to CCDV\@.
  But that means the generated case may not be single-peaked, as we
  would like it to be.

  We can work around this obstacle, by noting that the algorithm given
  in \citet{fal-hem-hem-rot:j:single-peaked-preferences} for the
  single-peaked CCDV approval-voting case in fact does a bit more than
  is claimed there.  It is easy to see, looking at that paper, that it
  in effect gives a polynomial-time algorithm for the following
  problem: Given an instance of CCDV, and a societal ordering, and
  given that in the instance's voter set each voter has an extra bit
  specifying whether the voter is deletable or is not deletable, and
  given that every voter that is specified as being deletable must be
  consistent with the societal ordering (but voters specified as being
  not deletable are not required to be consistent with the societal
  ordering---they may be mavericks), is there a set of at most $K$
  ($K$ being part of the input) deletable voters such that if we
  delete them our preferred candidate $p$ is a winner?  The fact that
  the paper implicitly gives such an algorithm is clear from that
  paper, since regarding the ``deleting voters'' actions described on
  its page 96 we can choose to allow them only on the deletable voters,
  and the \citet{fal-hem-hem-rot:j:single-peaked-preferences}
  algorithm's correctness in our case hinges (assuming that
  nondeletable voters are indeed nondeletable) just on the fact that
  the deletable voters all respect the societal ordering.  (Once we
  allow a deletable/nondeletable flag, we could in fact demaverickify
  all the remaining mavericks, and then flag all the 1-approval-each
  voters added by that demaverickification as being nondeletable,
  but there is no need to do any of that.  Doing it requires the
  deletable/nondeletable flag, and as just noted, if one has that
  flag, one can outright tolerate (nondeletable) mavericks.)

  So, we have noted a polynomial-time algorithm that, while not stated
  as their theorem, is a corollary to their theorem's proof---i.e.,
  the proof of the CCDV result that in~\citep{fal-hem-hem-rot:j:single-peaked-preferences} appears on that
  paper's page 93.  With this in hand, we can handle our
  nearly-single-peaked CCDV case using the same basic approach we used
  for CCAV, as naturally modified for the CCDV case.  In particular,
  we again polynomial-time disjunctive truth-table reduce to a problem
  known to be in polynomial time---in this case, CCDV for approval
  voting over single-peaked societies, with a deletable/nondeletable
  flag for each voter, and with all deletable voters having to be
  nonmavericks, which was argued above to be in polynomial time.  Our
  reduction is that (after checking that the input election is
  syntactically correct and does not have illegally many mavericks)
  for each subset of the mavericks that is of cardinality at most $K$
  (the input bound on the number of voters to delete), we delete those
  $K$ mavericks, then we decrement $K$ by the cardinality of the
  subset, then we mark all the nonmaverick voters as deletable, and
  mark each remaining maverick as nondeletable.
 This
 approach works for essentially the same reason as the CCAV case, and
 as in that case, we can get not merely a yes/no answer, but can even
 for the yes cases produce
  a successful control action.
\end{proof}

\begin{theorem4x4new}
CCAV and CCDV for Condorcet
elections over $f(\cdot)$-maverick-SP societies are each 
in $\NONDETTIME[ f(\mathrm{ProblemInputSize}),\, \poly]$.
For CCAV, the complexity remains in 
$\NONDETTIME[ f(\mathrm{ProblemInputSize}),\, \poly]$
even for the case where no limit is imposed on the 
number of mavericks in the initial voter set, and the
number of mavericks in the set of potential
additional voters is $f(\cdot)$-bounded 
(in the overall problem input size).
\end{theorem4x4new}
\begin{proof}
  Let us handle the CCAV case first. We assume we are in the more
  general setting where there is no limit on the number of mavericks
  in the initial voter set. Let $(C,V,W,p,K)$ be our input instance of
  CCAV for Condorcet and let $\sord$ be the societal axis.  We assume
  that the candidate set $C$ is of the form $C = \{b_{m'}, \ldots,
  b_1,p,c_1, \ldots, c_{m''}\}$ and $b_{m'} \sord \cdots \sord b_1
  \sord p \sord c_1 \sord \cdots \sord c_{m''}$ holds. We partition the voters in
  $W$ into fours groups, $W_\ell$, $W_r$, $W_p$, and $W_m$:
  \begin{enumerate}
  \item $W_\ell$ contains those voters from $W$ who are not mavericks and
    whose most preferred candidate $c$ is such that $c \sord p$
    (intuitively, these are the voters whose top choice is ``to the
    left of $p$'').
  \item $W_r$ contains those voters from $W$ who are not mavericks and 
    whose most preferred candidate $c$ is such that $p \sord c$
    (intuitively, these are the voters whose top choice is ``to the
    right of $p$'').
  \item $W_p$ contains those voters from $W$ whose most preferred
    candidate is $p$.
  \item $W_m$ contains the remaining voters from $W$ (i.e., $W_m$
    contains those mavericks who do not rank $p$ first; note that
    there are at most $f(\mathrm{ProblemInputSize})$ voters in $W_m$).
  \end{enumerate}

  Voters in $W_\ell$ (in $W_r$) have an interesting structure of their
  preference orders. For each $v \in W_\ell$ (each $v \in W_r$), due
  to his or her single-peakedness, there is a positive integer $i$
  such that $v$ prefers each candidate in $\{b_1, \ldots, b_i\}$ to
  $p$ to each of the remaining candidates ($v$ prefers each candidate
  in $\{c_1, \ldots, c_i\}$ to $p$ to each of the remaining
  candidates). Thus, we can conveniently sort the voters in $W_\ell$
  (the voters in $W_r$) in increasing order of the cardinalities
  of the sets of candidates they prefer to $p$.

  Our (nondeterministic) algorithm works as follows. First, we add
  $\max(\|W_p\|,K)$ voters from $W_p$. If that makes $p$ a Condorcet
  winner, we accept.  Otherwise, we set $K' = K - \max(\|W_p\|,K)$.
  If $K' = 0$, we reject. Next, using (at most)
  $f(\mathrm{ProblemInputSize})$ nondeterministic binary decisions,
  for each voter $v$ in $W_m$ we decide whether to add $v$ to the
  election or not. Let $M$ be the number of voters we add in this
  process. If $M > K'$ then we reject (on this computation path) and
  otherwise we set $K'' = K' - M$.  Then, we execute the following
  algorithm:
  \begin{enumerate}
  \item For each two nonnegative integers $K_\ell$ and $K_r$ such that
    $K_\ell + K_r \leq K''$, execute the following steps.
    \begin{enumerate} 
    \item Add $K_\ell$ voters from $W_\ell$ to the election (in the order
      described one paragraph above).
    \item Add $K_r$ voters from $W_r$ to the election (in the order
      described one paragraph above).
    \item Check if $p$ is the Condorcet winner of the resulting
      election.  If so, accept. Otherwise, undo the adding of the
      voters from the two preceding steps.
    \end{enumerate}  
  \item If we have not accepted until this point, reject on this
    computation path.
  \end{enumerate}
  It is easy to verify that this algorithm indeed runs in polynomial
  time (given access to $f(\mathrm{ProblemInputSize})$
  nondeterministic steps). Its correctness follows naturally from the
  observations regarding the preference orders of voters in $W_\ell$
  and $W_r$ (it is clear that, given $K_\ell$ and $K_r$, our algorithm adds
  $K_\ell$ voters from $W_\ell$ in an optimal way and adds $K_r$ voters from
  $W_r$ in an optimal way).

  Let us now turn to the case of CCDV. The algorithm is very similar.
  Let $(C,V,p,K)$ be our input instance.  We assume that the candidate
  set $C$ is of the form $C = \{b_{m'}, \ldots, b_1,p,c_1, \ldots,
  c_{m''}\}$ and $b_{m'} \sord \cdots \sord b_1 \sord p \sord c_1 \sord \cdots
  \sord c_{m''}$ holds. We partition the voters in $V$ into fours
  groups, $V_\ell$, $V_r$, $V_p$, and $V_m$:
  \begin{enumerate}
  \item $V_\ell$ contains those voters from $V$ who are not mavericks and
    whose most preferred candidate $c$ is such that $c \sord p$
    (intuitively, these are the voters whose top choice is ``to the
    left of $p$'').
  \item $V_r$ contains those voters from $V$ who are not mavericks and
    whose most preferred candidate $c$ is such that $p \sord c$
    (intuitively, these are the voters whose top choice is ``to the
    right of $p$'').
  \item $V_p$ contains those voters from $V$ whose most preferred
    candidate is $p$.
  \item $V_m$ contains the remaining voters from $V$ (i.e., $V_m$
    contains those mavericks who do not rank $p$ first; note that
    there are at most $f(\mathrm{ProblemInputSize})$ voters in $V_m$).
  \end{enumerate}
  For each $v \in V_\ell$ (each $v \in V_r$), due to his or her
  single-peakedness, there is a positive integer $i$ such that $v$
  prefers each candidate in $\{b_1, \ldots, b_i\}$ to $p$ to each of
  the remaining candidates ($v$ prefers each candidate in $\{c_1,
  \ldots, c_i\}$ to $p$ to each of the remaining candidates). Thus, we
  can conveniently sort the voters in $W_\ell$ (the voters in $W_r$)
  in decreasing order of the cardinalities of the sets of
  candidates they prefer to $p$.

  It is clear that we should never delete voters from $V_p$.  Our
  (nondeterministic) algorithm proceeds as follows. First, for each
  voter $v$ in $V_m$ we make a nondeterministic decision whether to
  delete $v$ from the election or not. Let $M$ be the number of voters
  we delete in this process. If $M > K$ then we reject on this
  computation path and otherwise we set $K' = K - M$. Next, we
  execute the following algorithm:
  \begin{enumerate}
  \item For each two nonnegative integers $K_\ell$ and $K_r$ such that $K_\ell+K_r \leq K'$,
    execute the following steps:
    \begin{enumerate}
    \item Delete $K_\ell$ voters from $V_\ell$ (in the order described one
      paragraph above).
    \item Delete $K_r$ voters from $V_r$ (in the order described one
      paragraph above).
    \item Check if $p$ is the Condorcet winner of the resulting
      election.  If so, accept. Otherwise, undo the deleting of the
      voters from the two preceding steps.
    \end{enumerate}
  \item If we have not accepted so far, reject on this computation
    path.
  \end{enumerate}
  Correctness and polynomial running time of the algorithm (given
  access to $f(\mathrm{ProblemInputSize})$ nondeterministic steps)
  follow analogously as in the CCAV case.
\end{proof}

\begin{theorem4x4}
For each $k$,
CCAC and CCDC for plurality
over $k$-maverick-SP societies are
  in $\p$.
\end{theorem4x4}

\begin{proof}
  The main idea of our proof is analogous to that of the proof of
  Theorem~\ref{t:stcp} but the details of demaverickification are
  different and, as a result, we can only handle a constant number of
  mavericks. We handle the CCAC case first.

  Let $I = (C,A,V,p,K)$ be our input instance of CCAC for plurality and
  let $\sord$ be the societal axis. Let $k'$ be the number of
  mavericks in $V$ ($k' \leq k$) and let $M = \{m_1, \ldots,
  m_{k'}\}$ be the subcollection of $V$ containing exactly these $k'$
  maverick voters. Our algorithm proceeds as follows:
  \begin{enumerate}
  \item\label{s:main-loop} For each vector $B = (b_1, \ldots, b_{k'}) \in (C \cup A)^{k'}$ of
    candidates execute the following steps (intuitively, we intend to
    enforce that candidates $b_1, \ldots, b_{k'}$ are top-ranked
    candidates of voters in $M$ and that it is impossible to change
    the top-ranked candidates of voters in $M$ by adding other
    candidates).
    \begin{enumerate}
    \item\label{s:a} If for any voter $m_i$, $1 \leq i \leq k'$, it
      holds that $m_i$ prefers some candidate in $(C \cup \{b_j
      \mid 1 \leq j \leq k'\}) - \{b_i\}$ to $b_i$ then drop this $B$
      and return to Step~\ref{s:main-loop}.  (This condition
      guarantees that it is possible to ensure, via adding candidates from
      $A$, that for each $i$, $1 \leq i \leq k'$, voter $m_i$ ranks
      candidate $b_i$ first among the participating candidates.)
    \item Set $C' = C \cup \{ b_j \mid 1 \leq j \leq k', b_j \in A \}$.
    \item Set $A' = (A - \{b_j \mid 1 \leq j \leq k', b_j \in A \}) - \{
      a \in A \mid$ some voter $m_i$, $1 \leq i \leq k'$, prefers $a$
      to $b_i \}$.
    \item Set $K' = K - \| \{ b_j \mid 1 \leq j \leq k', b_j \in A \}\|$. If $K' < 0$
      then drop this $B$ and return to Step~\ref{s:main-loop}.
    \item\label{s:e} Form a voter collection $V'$ that is identical to
      $V$ except that we restrict voters' preferences to candidates in
      $C' \cup A'$ and for each voter $m_i$, $1 \leq i \leq k'$, we
      replace $m_i$'s preference order with an easily-computable
      preference order over $C' \cup A'$ that ranks $b_i$ first and is
      single-peaked with respect to $\sord$.
    \item Using the polynomial-time algorithm of
      \citet{fal-hem-hem-rot:j:single-peaked-preferences}, check if
      $(C',A',V',p,K')$ is a yes instance of CCAC for
      single-peaked plurality elections with societal axis $\sord$. If
      so, accept.
    \end{enumerate}
    \item If the algorithm has not accepted so far, reject.
  \end{enumerate}

  It is easy to verify that the above algorithm indeed runs in
  polynomial time: There are exactly $\|C \cup A\|^{k'}$ choices of
  vector $B$ to test and for each fixed $B$ each step can clearly be
  performed in polynomial time. It remains to show that the algorithm
  is correct.

  Let us assume that $I$ is a yes instance. We will show that
  in this case the algorithm accepts. Let $A''$ be a subset of $A$
  such that $\|A''\| \leq K$ and $p$ is a winner of election $E'' = (C
  \cup A'',V)$. Let $B'' = (b''_1, \ldots, b''_{k'})$ be the vector of
  candidates from $C \cup A''$ such that for each $i$, $1 \leq i \leq
  k'$, in $E''$ voter $m_i$ ranks $b''_i$ first. We claim that
  our algorithm accepts at latest when it considers vector
  $B''$. First, by our choice of $B''$ it is clear that in
  Step~\eqref{s:a} we do not drop $B''$. Let $C'$, $A'$, $K'$, and
  $V'$ be as computed by our algorithm for $B = B''$. By our choice of
  $B''$, it is clear that $A'' - \{b''_i | 1 \leq i \leq k'\}
  \subseteq A'$. Thus, there is a set $A''' \subseteq A'$ such that
  $\|A'''\| \leq K'$ and $p$ is a winner of election $(C' \cup
  A''',V)$. Further, since every voter $m_i$, $1 \leq i \leq k'$,
  prefers candidate $b''_i$ to all other candidates in $C' \cup A'$,
  it holds that $p$ is a winner of election $(C' \cup A''',V')$.  That
  is, $(C',A',V',p,K')$ is a yes instance.

  Similarly, it is easy to see that the construction of instances
  $(C',A',V',p,K')$, and in particular the construction of $V'$ in
  Step~\eqref{s:e}, ensures that if the algorithm accepts then $I$ is
  a yes instance. This completes the discussion of the CCAC
  case.

  \medskip

  Let us now move on to the case of CCDC. As in the case of CCAC, we
  will, essentially, reduce the problem to the case where all voters
  are single-peaked. However, we will need the following more general
  variant of the CCDC problem.
  \begin{definition}
    Let $R$ be an election system.  In the CCDC with restricted
    deleting problem for $R$, we are given an election $(C,V)$, a
    candidate $p \in C$, a set $F \subseteq C$ such that $p \in F$,
    and a nonnegative integer $K$. We ask if there is a set $C'
    \subseteq C$ such that (a) $\|C'\| \leq K$, (b) $C' \cap F =
    \emptyset$, and (c) $p$ is a winner of $R$ election $(C -
    C', V)$.
  \end{definition}
  That is, in CCDC with restricted deleting we can specify which
  candidates are impossible to delete. The following result is a
  direct corollary to the proof of
  \citet{fal-hem-hem-rot:j:single-peaked-preferences} that CCDC for
  plurality is in $\p$ for single-peaked electorates.
  \begin{observation}[Implicit in
    \citet{fal-hem-hem-rot:j:single-peaked-preferences}]\label{t:ccdc-restricted}
    CCDC with restricted deleting is in $\p$ for plurality over
    single-peaked electorates.
  \end{observation}

  Let $I = (C,V,p,K)$ be our input instance of CCDC for plurality and
  let $\sord$ be the societal axis. Let $k'$ be the number of
  mavericks in $V$ ($k' \leq k$) and let $M = \{m_1, \ldots, m_{k'}\}$
  be the subcollection of $V$ that contains these $k'$ maverick
  voters. Our algorithm works as follows:
  \begin{enumerate}
  \item\label{s:main-loop2} For each vector $B = (b_1, \ldots, b_{k'})
    \in C^{k'}$ of candidates execute the following steps
    (intuitively, we intend to enforce that candidates $b_1, \ldots,
    b_{k'}$ are top-ranked candidates of voters in $M$).
    \begin{enumerate}
    \item If for any voter $m_i$, $1 \leq i \leq k'$, it holds that
      $m_i$ prefers some candidate in $(C \cup \{b_j
      \mid 1 \leq j \leq k'\}) - \{b_i\}$ to $b_i$ then drop this $B$
      and return to
      Step~\ref{s:main-loop2}.  (This condition guarantees that it is
      possible to ensure, via deleting voters, that for each $i$, $1
      \leq i \leq k'$, voter $m_i$ ranks candidate $b_i$ first among
      the participating candidates.)
    \item Set $F' = \{b_i \mid 1 \leq i \leq k'\} \cup \{p\}$.
    \item Set $C' = C - \{ c \in C \mid$ there is an $i$, $1 \leq i
      \leq k'$ such that $m_i$ prefers $c$ to $b_i \}$.
    \item Set $K' = K - \|\{ c \in C \mid$ there is an $i$, $1 \leq i \leq
      k'$ such that $m_i$ prefers $c$ to $b_i \}\|$. If $K' < 0$ then
      drop this $B$ and return to Step~\ref{s:main-loop2}.
    \item
Form a voter collection $V'$ that is identical to
      $V$ except that we restrict voters' preferences to candidates in
      $C'$ and for each voter $m_i$, $1 \leq i \leq k'$, we replace
      $m_i$'s preference order with an easily-computable preference
      order over $C'$ that ranks $b_i$ first and is single-peaked with
      respect to $\sord$.
    \item Using Observation~\ref{t:ccdc-restricted}, check if
      $(C',V',p, F',K')$ is a yes instance of CCDC with
      restricted deleting for single-peaked plurality elections with
      societal axis $\sord$. If so, accept.
    \end{enumerate}
    \item If the algorithm has not accepted so far, reject.
  \end{enumerate}

  Using the same arguments as in the case of CCAC, we can see that
  this algorithm is both correct and runs in polynomial time.
\end{proof}

\subsection{Proofs of Theorems~\ref{t:plurality-swoon},~\ref{t:plurality-manymavericks}, and~\ref{t:dodgson-C-2}}

\begin{theorem4x5}
CCAC and CCDC for plurality
elections over swoon-SP societies are
  $\np$-complete.
\end{theorem4x5}

Theorem~\ref{t:plurality-swoon} implies
the classic result of \citet{bar-tov-tri:j:control} that CCAC and CCDC
are $\np$-complete
for plurality elections. However, the proofs of Bartholdi, Tovey, and
Trick cannot be used directly for swoon-SP societies.

\begin{proof}[Proof of Theorem~\ref{t:plurality-swoon}]
  We first consider the case of CCAC.
  We easily note that CCAC for plurality over swoon-SP societies is in
  $\np$. It remains to show that it is $\np$-hard and we do so by
  giving a reduction from X3C.  Let $I = (B, \calS)$ be our input X3C
  instance, where $B = \{b_1, \ldots, b_{3k}\}$ and $\calS = \{S_1,
  \ldots, S_n\}$. Without loss of generality, we assume that $k \geq 2$
  and $n \geq 4$. 
  For each $b_i \in B$, we set $\ell_i$ to be the number of sets in
  $\calS$ that contain $b_i$.

  We construct an election $E = (C \cup A, V)$, where $C = B \cup
  \{p,d\}$ is the set of registered candidates, $A = \{a_1, \ldots,
  a_n\}$ is the set of spoiler (unregistered) candidates, and $V$ is a
  collection of votes. Each candidate $a_i$ in $A$ corresponds to a set
  $S_i$ in $\calS$. We assume that the societal axis $\sord$ is $p
  \sord d \sord b_1 \sord \cdots \sord b_{3k} \sord a_1 \sord \cdots
  \sord a_n$. (Our proof works for any easily computable axis.) Collection
  $V$ contains the following $(6kn) + (n) + (2nk+k-n) + (2nk) +
  \left(\sum_{i=1}^{3k}\left(2nk+2k-2k\ell_i\right)\right)$ votes; each
  of the five parenthesized terms in this expression corresponds to an item in the description
  of votes below.  (For each vote we only specify up to two top-ranked
  candidates. Note that voters in swoon-SP societies can legally pick
  any two candidates to be ranked in the top two positions of their
  votes. This is so, because the top-ranked candidate can be chosen
  freely as the candidate to which the voter swoons, and the
  second-ranked candidate can be chosen to be the voter's peak in the
  societal axis. We assume that the remaining positions in each
  vote---irrelevant from the point of view of our proof---are filled
  in in an easily computable way consistent with the societal axis
  $\sord$. For example, each voter we describe below could rank
  candidates as follows: (a) in the first up to two positions of the
  vote he or she would rank the candidates as described below
  (appropriately choosing the candidate to which he or she swoons, and
  the candidate who takes the role of the voter's peak on the societal
  axis), (b) in the remaining positions the voter would first rank the
  remaining candidates ``to the left'' of the peak and then those ``to
  the right'' of the peak.)
  \begin{enumerate}
  \item For each set $S_j \in \calS$, for each $b_i \in S_j$, we have $2k$ 
    votes $a_j > b_i > \cdots$.
  \item For each set $S_j \in \calS$ we have a single vote $a_j > p > \cdots$.
  \item We have $2nk+k - n$ voters that rank $p$ first.
  \item We have $2nk$ voters that rank $d$ first.
  \item For each $b_i \in B$, we have $2nk+2k-2k\ell_i$ voters
   that rank $b_i$ first.
  \end{enumerate}
  We note that in election $(C,V)$ the scores of candidates 
  are as follows:
  \begin{enumerate}
  \item $p$ has $2nk+k$ points,
  \item $d$ has $2nk$ points, and
  \item each candidate $b_i \in B$ has $2nk+2k$ points.
  \end{enumerate}
  That is, the winners of plurality election $(C,V)$ are exactly the
  candidates in $B$. We claim that there is a set $A'$, $A' \subseteq
  A$, such that $\|A'\| \leq k$ and $p$ is a winner of plurality
  election $(C \cup A',V)$ if and only if $I$ is a yes instance
  of X3C (that is, if there exists a collection of exactly $k$ sets
  from $\calS$ that union to $B$; such a collection of sets is called
  an exact set cover of $B$).

  Let $A''$ be some subset of $A$. It is easy to see that in election
  $(C \cup A'',V)$, plurality scores of candidates are as follows: $p$
  has score $2nk + k - \|A''\|$, $d$ has score $2nk$, each candidate
  $b_i \in B$ has score $2nk+2k - 2k\|\{ a_j \in A'' \mid b_i \in
  S_j\}\|$, and each candidate $a_i \in A''$ has score $6k+1$. 

  Assume that $p$ is a winner of election $(C \cup A'',V)$.  Since
  $d$'s score is $2nk$ and $p$'s score is $2nk+k - \|A''\|$, it holds
  that $\|A''\| \leq k$. Further, for each $b_i \in B$ it holds that
  $b_i$'s score is no larger than that of $p$. It is easy to verify
  that this is possible only if $A''$ corresponds to an exact set
  cover of $B$ (the score of each of $3k$ candidates in $B$ has to be
  decreased and each $a_j \in A''$ corresponds to decreasing the score
  of exactly three candidates in $B$).

  On the other hand, if $A''$ corresponds to an exact cover of $B$,
  then $p$ is a winner of election $(C \cup A'',V)$. In such a case
  $\|A''\| = k$ and so the score of $p$ is $2nk$. Since each $a_j \in
  A''$ corresponds to a set $S_j \in \calS$ that contains three unique
  members of $B$, the score of each $b_i \in B$ is $2nk$. The score of
  $d$ is $2nk$ as well. Each $a_j \in A''$ has score $6k+1 < 2nk$
  (this is so because $n \geq 4$ and $k \geq 2$). The proof is
  complete.

  \medskip

  We now move on to the case of CCDC.  CCDC for plurality
  over swoon-SP societies is clearly in $\np$ and we focus on proving
  $\np$-hardness. We do so by giving a reduction from X3C.  Let $I =
  (B, \calS)$ be an X3C instance, where $B = \{b_1, \ldots, b_{3k}\}$
  and $\calS = \{S_1, \ldots, S_n\}$. Without loss of generality we
  assume that $k > 5$. We use societal axis $p \sord d \sord b_1 \sord
  \cdots \sord b_{3k} \sord a_1 \sord \cdots \sord a_n$.

  We construct an instance of CCDC for plurality as follows. Set $A =
  \{a_1, \ldots, a_n\}$ and let $E = (C,V)$ be an election, where $C = B
  \cup A \cup \{p\}$ and $V$ contains the following groups of votes
  (for each vote we only specify up to two top candidates and up to
  one ranked-lowest candidate; the reader can verify that using
  societal axis $\sord$ it is possible to create swoon-SP votes of the
  form we require).
  \begin{enumerate}
  \item For each $S_j \in \calS$ and for each $b_i \in S_j$, we have
    one vote $a_j > b_i > \cdots > p$.
  \item For each $S_j \in \calS$, we have one vote $a_j > p > \cdots$.
  \item For each $b_i \in B$, we have $k-1$ votes $b_i > \cdots > p$.
  \end{enumerate}
  In this election the candidates have the following scores:
  \begin{enumerate}
  \item $p$ has $0$ points,
  \item each $b_i \in B$ has $k-1$ points, and
  \item each $a_j$, $1 \leq j \leq n$, has $4$ points (note that $4 <
    k-1$).
  \end{enumerate}
  We claim that it is possible to ensure that $p$ is a winner of this
  election by deleting at most $k$ candidates if and only if $I$ is a
  yes instance of X3C.

  First, assume that $I$ is a yes instance of X3C and let $A'$
  be a subset of $A$ such that $\{S_i \mid a_i \in A'\}$ is an exact
  cover of $B$. It is easy to see that $p$ is a plurality winner of
  election $E' = (C - A',V)$: Compared to $E$, in $E'$ the
  score of $p$ increases by $k$, the score of each $b_i \in B$
  increases by $1$, and the scores of remaining members of $A$ do not
  change. Thus, $p$ and all members of $B$ tie for victory.

  On the other hand, assume that there exists a set $A'' \subseteq B
  \cup A$ of candidates, $\|A''\| \leq k$, such that $p$ is a winner
  of election $E'' = (C - A'',V)$. Since $\|A''\| \leq k$, there
  are at least $2k$ candidates from $B$ in $E''$ and so the score of
  $p$ in $E''$ has to be at least $k-1$, to tie with these
  candidates. However, the only way to increase $p$'s score to
  $k-1$ (or higher) by deleting at most $k$ candidates is to delete
  $k-1$ (or more) candidates from $A$.  Yet if we delete exactly
  $k-1$ candidates from $A$, then there is some candidate $b_i$ in the
  election whose score is at least $k$. Thus, $A''$ must contain
  exactly $k$ candidates from $A$. Deleting these candidates increases
  $p$'s score to be $k$. To ensure that the scores of the candidates
  in $B$ do not exceed $k$, we must ensure that $A''$ corresponds to
  an exact cover of $B$ by sets from $\calS$. This completes the
  proof.
\end{proof}

By a simple extension of the above proof, we can also show that
allowing the number of mavericks to be some root of the input size
cannot be handled either.

\begin{theorem4x6}
For each $\epsilon > 0$,
CCAC and CCDC for plurality
elections over $I^\epsilon$-maverick-SP societies
are $\np$-complete, where $I$ denotes the input size.
\end{theorem4x6}
\begin{proof}
  If $\epsilon \geq 1$ then all voters can be mavericks and the
  theorem certainly holds (because CCAC and CCDC are $\np$-complete
  for plurality with unrestricted votes).  Let us consider the case
  when $\epsilon$ is strictly between $0$ and $1$. In this case we can
  adapt the proof of Theorem~\ref{t:plurality-swoon} by including an
  appropriate number of padding voters.

  Let us first handle the case of CCAC. Let $(B,\calS)$ be an instance
  of X3C where $B = \{b_1, \ldots, b_{3k}\}$ and $\calS = \{S_1,
  \ldots, S_n\}$, and let $(C,A,V,p,K)$ be the instance of CCAC for
  plurality produced by the reduction in the proof of
  Theorem~\ref{t:plurality-swoon}. Let $\sord$ be the societal axis
  used in the proof of Theorem~\ref{t:plurality-swoon}. By definition,
  $(B,\calS)$ is a yes instance of X3C if and only if
  $(C,A,V,p,K)$ is a yes instance of CCAC for
  plurality. However, of course, we have no guarantee that $V$
  contains at most $I^\epsilon$ mavericks (with respect to $\sord$),
  where $I$ denotes the input size of $(C,A,V,p,K)$. Yet it is easy
  to verify that for each positive integer $t$, $(C,A,V,p,K)$ is a
  yes instance of CCAC for plurality if and only if $(C,A,V \cup
  V'_t,p,K)$ is a yes instance of the same problem, where $V'_t$ is a
  collection of $t$ blocks of votes that each contain the following
  $3k+2$ votes:
  \begin{enumerate}
  \item For each $i$, $1 \leq i \leq 3k$, there is a single vote that
    is single-peaked with respect to $\sord$ and ranks $b_i$ first and
    $p$ last\footnote{For CCAC it is not even necessary to require
      that $p$ is ranked last. However, we will also use $V'_t$ for
      the CCDC case, where it \emph{is} necessary.} (note that, by our
    choice of $\sord$ in the proof of Theorem~\ref{t:plurality-swoon},
    such a vote exists).
  \item There is a single vote that is single-peaked with respect to
    $\sord$ and ranks $p$ first.
  \item There is a single vote that is single-peaked with respect to
    $\sord$ and ranks $d$ first and $p$ last.
  \end{enumerate}
  It is easy to see that by choosing a large enough value of $t$ (but
  polynomially bounded in $I^{\frac{1}{\epsilon}}$), it is possible to
  form an instance $(C,A,V \cup V'_t,p,K)$, whose encoding size is $I'$,
  that is a yes instance of CCAC for plurality if and only if
  $(B,\calS)$ is a yes instance of plurality, and which
  contains at most $I'^\epsilon$ mavericks with respect to the
  societal axis $\sord$ (namely, the voters in $V$). This proves our
  theorem for the CCAC case.

  Essentially the same proof approach works for the CCDC case. The
  crucial observation here is that the proof of the CCDC case of
  Theorem~\ref{t:plurality-swoon} ensures that deleting candidates
  outside of the set $A$ is never a successful strategy. Adding the
  voters $V'_t$ does not affect this observation because in all votes
  in $V'_t$ candidate $p$ is either ranked first or ranked last
  (however, of course, for the case of CCDC each of the $t$ blocks of
  votes in $V'_t$ contains only $3k+1$ votes; the vote that ranks $d$
  first and $p$ last is not included, and the candidate $d$ does not
  occur in any of the other votes).
\end{proof}

\begin{theorem4x8}
CCAC and CCDC for plurality
elections are $\np$-complete for 
Dodgson$_{m-2}$-SP societies and for 
PerceptionFlip$_{m-2}$-SP societies, where
$m$ is the total number of candidates involved
in the election.
\end{theorem4x8}

It is easy to see that Theorem~\ref{t:dodgson-C-2} is a simple
corollary to the proof of Theorem~\ref{t:plurality-swoon}. If $m$ is
the total number of candidates involved in the election then both in
Dodgson$_{m-2}$-SP societies and in PerceptionFlip$_{m-2}$-SP
societies the voters can legally rank any two candidates on top of
their votes (see lemma below). Further, the societal axis in the CCDC
part of the proof of Theorem~\ref{t:plurality-swoon} is such that the
voters can easily rank $p$ last if need be. This is all that we need
for the proof of Theorem~\ref{t:plurality-swoon} to work for the case
of Theorem~\ref{t:dodgson-C-2}.

\begin{lemma}
  Let $C = \{c_1, \ldots, c_m\}$ be a set of candidates, $m \geq 2$,
  and let $\sord$ be a societal axis over $C$ such that $c_1 \sord c_2
  \sord \cdots \sord c_m$. For each two candidates $c_i, c_j \in C$,
  there exist two preference orders of the form $c_i > c_j > \cdots$
  such that the first is nearly single-peaked in the sense of
  Dodgson$_{m-2}$-SP societies and the second one is nearly
  single-peaked in the sense of PerceptionFlip$_{m-2}$-SP
  societies. Further, if $i \neq 1$ and $j \neq 1$, it is possible to
  ensure that these preference orders rank $c_1$ last.
\end{lemma}
\begin{proof}
  Let $c_i$ and $c_j$ be two arbitrary, distinct candidates in $C$.
  We first consider the case of Dodgson$_{m-2}$-SP societies.  Let
  $>'$ be an arbitrary preference order that is single-peaked with
  respect to $\sord$ and that ranks $c_i$ first (and $c_1$ last, if $i
  \neq 1$ and $j \neq 1$). We obtain $>$ from $>'$ by shifting $c_j$
  forward in $>'$ to the second position (that is, just below
  $c_i$). It is easy to see that this requires at most $m-2$
  swaps.

  For the case of PerceptionFlip$_{m-2}$-SP societies, note that
  by using at most $m-2$ swaps of adjacent candidates, it is
  possible to obtain a societal axis $\sord'$ from $\sord$ where
  candidates $c_i$ and $c_j$ are adjacent (and where, if $i \neq 1$
  and $j \neq 1$, it still holds that $c_1 \sord c_k$ for each $k$, $2
  \leq k \leq m$). Clearly, there is a preference order $>$ that is
  single-peaked with respect to $\sord'$ and that ranks $c_i$ first
  and $c_j$ second (and $c_1$ last, if $i \neq 1$ and $j \neq 1$).
\end{proof}

\subsection{Proof of Theorem~\ref{t:klocal}}

\begin{theoremklocal}
For each constant $k$,
CCAC and CCDC for plurality
elections are in $\p$ for $k$-local elections. 
\end{theoremklocal}

We first give a polynomial-time algorithm for CCAC for plurality
$k$-local elections.  The main idea of our algorithm is the following.
Let $p$ be the candidate whose victory we want to ensure in our input
$k$-local instance of plurality CCAC.  We first add up to $2k$
candidates so that the score of $p$ is fixed, and then we run a
dynamic programming algorithm that ensures that no candidate has score
higher than this fixed score of $p$.  Of course, we do not know which
candidates to add in the first part of the algorithm, so we perform an
exhaustive search (since $k$ is a constant, it is possible to perform
such a search in polynomial time). We will first describe the dynamic
programming algorithm in Lemma~\ref{l:k-local-dynamic} and then we
will provide the main algorithm.  Before we proceed with this plan, we
need to provide some additional notation.

Let $E = (C \cup A, V)$ be a $k$-local plurality election, where we
interpret $C$ as the registered candidates and $A$ as the spoiler
candidates. Let $\sord$ be the societal axis for $E$. We rename the
candidates so that $D = C \cup A = \{d_1, \ldots, d_m\}$ and $d_1
\sord d_2 \sord \cdots \sord d_m$. For each set $B \subseteq D$, we
define $\LEFT(B)$ to be the minimal (leftmost) element of $B$ with
respect to $\sord$ and $\RIGHT(B)$ to be the maximal (rightmost)
element of $B$ with respect to $\sord$. For each $d_i \in D$ we define
$\calS(d_i)$ to be the family of sets $\{ N(\sord,C \cup A', d_i, k) \mid A'
\subseteq A, d_i \in C \cup A'\}$. The reader can verify that each
$\calS(d_i)$ contains a number of sets that is at most polynomial in
$\|C \cup A\|^k$ and that each $\calS(d_i)$ is easily computable
(to compute $\calS(d_i)$ it suffices to consider sets $A'$ of
cardinality at most $2k+1$).

\begin{lemma}\label{l:k-local-dynamic}
  Let $E = (C \cup A, V)$ be a plurality election, where $C = \{c_1, \ldots,
  c_{m'}\}$ and $A = \{a_1, \ldots, a_{m''}\}$, such that $E$ is
  $k$-local for some positive integer $k$. There exists an algorithm
  that given election $E$, integer $k$, societal axis $\sord$ with
  respect to which $E$ is $k$-local, and a nonnegative integer $t$
  outputs the cardinality of a smallest (in terms of cardinality) set
  $A' \subseteq A$ such that the plurality scores of all candidates in
  election $(C \cup A',V)$ are at most $t$, or indicates that no such
  set $A'$ exists.  This algorithm runs in time polynomial with
  respect to $(\|C \cup A\|+\|V\|)^k$.
\end{lemma}
\begin{proof}
  The proof of this lemma is a much extended version of the proof of
  Lemma~3.7 of~\citep{fal-hem-hem-rot:j:single-peaked-preferences}.
  Let the notation be as in the statement of the lemma. We assume that
  $C$ is nonempty.

  We let $D = C \cup A$ and, without loss of generality, we rename the
  candidates so that $D = \{d_1, \ldots, d_m\}$, where $m = m'+m''$,
  and $d_1 \sord d_2 \sord \cdots \sord d_m$. Without loss of
  generality, we assume that $d_1, d_m \in C$ (if this was not the
  case, we could extend $C$ to include two additional candidates,
  ranked last by all voters, without destroying $k$-locality of the
  election).
  
  For each $d_i \in D$ and each $D' \in \calS(d_i)$ we define $f(d_i,
  D')$ to be the cardinality of a smallest (with respect to
  cardinality) set $A' \subseteq A$ such that:
  \begin{enumerate}
  \item For each candidate $d_j \in C \cup A'$ such that $j \leq i$ it
    holds that $\score_{(C \cup A',V)}(d_j) \leq t$.
  \item If $d_{j'} = \LEFT(D')$ and $d_{j''} = \RIGHT(D')$ then $D' =
    (C \cup A') \cap \{d_{j'}, d_{j'+1}, \ldots, d_{j''}\}$.
(Since $d_1,d_m \in C$, this is equivalent to
$D' = N(\sord,C\cup A',d_i,k)$.)
  \end{enumerate}
  If such an $A'$ does not exist, then we set $f(d_i, D') = \infty$.  

  Let $d_i$ be a candidate in $D$ and let $D'$ be a member of
  $\calS(d_i)$. Intuitively, $D'$ describes the intended $k$-radius
  neighborhood of $d_i$. Function $f(d_i, D')$ tells us how many
  candidates from $A$ we need to add to election $(C,V)$ so that in
  the resulting election the $k$-radius neighborhood of $d_i$ is
  exactly $D'$ (which fixes the score of $d_i$), the score of $d_i$ is
  at most $t$, and the scores of candidates preceding $d_i$ (in terms
  of $\sord$) also are at most $t$.

Since $d_m \in C$, it is easy to verify that our algorithm
  should output $\min\{ f(d_m, D') \mid  D' \in \calS(d_m) \}.$
  Thus, in the rest of the proof we describe how to compute $f$ using
  dynamic programming.

  It is easy to see that for each $D' \in \calS(d_1)$ it is possible
  to directly compute the value $f(d_1, D')$. To compute $f(d_i, D')$
  for arbitrary $d_i \in D$, $D' \in \calS(d_i)$, we use the
  following, natural to derive, recursive relation. Let us fix some
  $d_i \in D$, $i > 1$, and $D' \in \calS(d_i)$.
Let $j = \max\{j' \mid j' < i \mbox{ and } d_{j'} \in D'\}$. 
We observe that:
  \begin{align*}
    f(d_i,& D') = \min\{f(d_{j}, D'') \mid
 D'' \in \calS(d_{j}) \mbox{ and}\\
    &\score_{(D',V)}(d_i) \leq t \mbox{ and } D' \cup \{\LEFT(D'')\} = D'' \cup \{\RIGHT(D')\} \} +\chi_{A}(\RIGHT(D')),
  \end{align*}
  where $\chi_{A}$ is the characteristic function of $A$.
(Note that if $\RIGHT(D') \in D''$ for $D'' \in \calS(d_{j})$,
then $\RIGHT(D') = d_m$ and so $\chi_{A}(\RIGHT(D')) = 0$.)
 Using
  standard dynamic programming techniques we can thus compute $f$ in
  time polynomial in $(\|C \cup A\|+\|V\|)^k$.
\end{proof}

We now prove that CCAC for plurality $k$-local elections is in $\p$.

\begin{lemma}
  For each fixed $k$, CCAC for $k$-local plurality elections, where
  the societal axis $\sord$ is given, is in $\p$.
\end{lemma}
\begin{proof}
  Our input instance contains the following elements: (a) an election
  $E = (C \cup A,V)$, where $C = \{c_1, \ldots, c_{m'}\}$ are the
  registered candidates and $A = \{a_1, \ldots, a_{m''}\}$ are the
  spoiler candidates; (b) a candidate $p$ in $C$ (of course, $p$ is
  one of the $c_i$'s but we assign him or her also this special name);
  (c) a nonnegative integer $K$; (d) a societal axis $\sord$ over $C
  \cup A$. We ask if there exists a set $A' \subseteq A$ such that:
  \begin{enumerate}
  \item $\|A'\| \leq K$, and
  \item $p$ is a winner of plurality election $(C \cup A',V)$.
  \end{enumerate}

  We rename the candidates so that $D = C \cup A = \{d_1, \ldots,
  d_m\}$, where $m = m'+m''$ and $d_1 \sord d_2 \sord \cdots \sord
  d_m$. Let $w$ be an integer such that $p = d_w$.  

  The idea of our algorithm is the following: Fix the $k$-radius
  neighborhood of $p$ (so that we fix $p$'s score) and ensure---using
  the algorithm from Lemma~\ref{l:k-local-dynamic}---that the
  remaining candidates have no more points than $p$ has. 

  Our algorithm works as follows.  For each possible $k$-radius
  neighborhood $D'$ of $p$ (i.e., for each $D' \in \calS(p)$) we
  execute the following steps.
  \begin{enumerate}
  \item Set $K_p = \|D' \cap A\|$. ($K_p$ is the number of candidates we
    need to add to ensure that $p$ has exactly $k$-radius neighborhood
    $D'$.)
  \item Set $t = \score_{(D',V)}(p)$ (by $E$'s $k$-locality, $t$ is the
    score of $p$ in any election where the $k$-radius neighborhood of $p$
    is $D'$).
  \item Check how many candidates are needed to ensure that candidates
    ``to the left'' of $p$ do not beat $p$):
    \begin{enumerate}
    \item Set $j'$ to be such that $d_{j'} = \LEFT(D')$.
    \item Set $C_{\mathit{left}} = (\{d_1, \ldots, d_{j'-1}\} \cap C)
      \cup D'$ and set $A_{\mathit{left}} = \{d_1, \ldots, d_{j'-1}\} \cap
      A$.  ($C_{\mathit{left}}$ is the set of all registered
      candidates ``to the left'' of $D'$, union $D'$ (we treat $D'$ as
      already fixed); $A_{\mathit{left}}$ is the set of spoiler
      candidates to the left of $D'$.)
    \item Using Lemma~\ref{l:k-local-dynamic} compute the minimal
      number of candidates from $A_{\mathit{left}}$ that need to be
      added to election $(C_{\mathit{left}},V)$ so that each candidate
      in the resulting election has score at most $t$. Call this number
      $K_{\mathit{left}}$. If it is impossible to achieve the desired
      effect, drop this $D'$.
    \end{enumerate}    

  \item Check how many candidates are needed to ensure that candidates
    ``to the right'' of $p$ do not beat $p$):
    \begin{enumerate}
    \item Set $j''$ to be such that $d_{j''} = \RIGHT(D')$.
    \item Set $C_{\mathit{right}} = (\{d_{j''+1}, \ldots, d_{m}\} \cap
      C) \cup D'$ and set $A_{\mathit{right}} = \{d_{j''+1}, \ldots,
      d_{m}\} \cap A$. ($C_{\mathit{right}}$ is the set of all
      registered candidates ``to the right'' of $D'$, union $D'$ (we
      treat $D'$ as already fixed); $A_{\mathit{right}}$ is the set of
      spoiler candidates ``to the right'' of $D'$.)
    \item Using Lemma~\ref{l:k-local-dynamic} compute the minimal
      number of candidates from $A_{\mathit{right}}$ that need to be
      added to election $(C_{\mathit{right}},V)$ so that each
      candidate in the resulting election has score at most $t$. Call
      this number $K_{\mathit{right}}$. If it is impossible to achieve
      the desired effect, drop this $D'$.
    \end{enumerate}
   \item If $K_p + K_{\mathit{left}} + K_{\mathit{right}} \leq K$ then accept.
  \end{enumerate}
  If the above procedure does not accept for any $D'$ then reject.

  By Lemma~\ref{l:k-local-dynamic} and the fact that $k$ is a fixed
  constant, it is easy to see that this algorithm works in polynomial
  time. The correctness is easy to observe as well.
\end{proof}

We now move on the the case of CCDC for plurality $k$-local elections.

\begin{lemma}
  For each fixed $k$, CCDC for $k$-local plurality elections, where
  the societal axis $\sord$ is given, is in $\p$.
\end{lemma}
\begin{proof}
  Let $E = (C,V)$ be our input election, $p$ be the preferred
  candidate, and $K$ be a nonnegative integer. Our goal is to determine
  if it is possible to ensure that $p$ is a winner by deleting at most $K$
  candidates. Let $\sord$ be the input societal axis with respect to
  which $E$ is $k$-local.

  We rename the candidates in $C$ so that $C = \{\ell_{m'}, \ldots,
  \ell_1, p, r_1, \ldots, r_{m''}\}$ and $\ell_{m'} \sord \cdots \sord
  \ell_1 \sord p \sord r_1 \sord \cdots \sord r_{m''}$. Recall that by
  definition of $k$-local plurality elections, the score of $p$
  depends only on the presence of $k$ candidates ``to the left of
  $p$'' and $k$ candidates ``to the right of $p$'' (with respect to
  $\sord$). Our algorithm works as follows:
  \begin{enumerate}
  \item For each size-$\min(k,m')$ subset $L$ of $\{\ell_1, \ldots, \ell_{m'}\}$
    and each size-$\min(k,m'')$ subset $R$ of $\{r_1, \ldots, r_{m''}\}$ execute
    the following steps:
    \begin{enumerate}
    \item Let $i$ be the largest integer such that $\ell_i \in L$ and
      let $j$ be the largest integer such that $r_j \in R$.
    \item Let $D = \{\ell_i, \ldots, \ell_1, r_1, \ldots r_j\}
      - (L \cup R)$ (at this point, intuitively, $D$ is the
      unique smallest set of candidates that one has to delete
      from $C$ to ensure that the $k$-radius neighborhood of $p$ is
      exactly $L \cup R$).
    \item Execute the following loop: If there is a candidate $c \in C
      - D$, $c \neq p$, such that the score of $c$ in $(C
      - D,V)$ is higher than that of $p$, then add $c$ to $D$.
    \item If $\|D\| \leq K$ then accept.
    \end{enumerate}
  \item Reject.
  \end{enumerate}
  Since $k$ is a constant, there are only polynomially many pairs of
  sets $L$ and $R$ to try. Thus, it is easy to see that the algorithm
  runs in polynomial time. To see the correctness, it suffices to note
  the following two facts. First, the score of $p$ depends only on the
  $k$-radius neighborhood of $p$. Second, it is impossible to decrease
  a score of a candidate by deleting (other) candidates, so if for a
  given $k$-radius neighborhood of $p$ some candidates still have
  score higher than $p$, the only way to ensure that they do not preclude
  $p$ from winning is by deleting them.\footnote{Note that in the process of
    doing so we might change the $k$-radius neighborhood of $p$, but
    that does not affect the correctness of the algorithm.}
\end{proof}

\subsection{Proof of Theorem~\ref{t:plurality-caved}}
\newcommand{\calD}{{{\mathcal{D}}}}

The following lemma will be very useful in proving
Theorem~\ref{t:plurality-caved}.

\begin{lemma}\label{l:s-c}
  Let $E = (C,V)$ be an election where $C = \{c_1, \ldots, c_m\}$ and
  where voters in $V$ are single-caved with respect to societal axis
  $c_1 \sord c_2 \sord \cdots \sord c_m$. Then each voter in $V$ ranks
  first either $c_1$ or $c_m$.
\end{lemma}
\begin{proof}
  Assume for the sake of contradiction that there is $c_i \in C$, $i
  \notin \{1,m\}$, such that some voter $v$ in $V$ ranks $c_i$ first.
  However, it holds that $c_1 \sord c_i \sord c_m$. Since, by
  assumption, $v$ prefers $c_i$ to $c_1$, by definition of
  single-cavedness it holds that $v$ prefers $c_m$ to $c_i$. This is a
  contradiction.
\end{proof}

With Lemma~\ref{l:s-c} in hand, we can now prove
Theorem~\ref{t:plurality-caved}.

\begin{theorem4x9}
CCAC and CCDC for plurality
elections are in $\p$ for 
single-caved societies.
\end{theorem4x9}
\begin{proof}%
  Let us consider the CCAC case first. Let $(C,A,V,p,K)$ be our input
  instance of CCAC for plurality, where votes in $V$ are single-caved
  with respect to a given societal axis $\sord$.
  We assume that
  $\|C\| \geq 2$ (otherwise $p$, the only candidate, is already a
 winner).

  Let $a$ be some candidate in $A$. We claim that if $p$ is not a
  winner of election $E = (C,V)$ then $p$ is not a winner of election
  $E' = (C \cup \{a\}, V)$.  First, adding $a$ cannot increase $p$'s
  plurality score. Thus, if $p$'s plurality score in $E$ is $0$ then
  it is $0$ in $E'$ as well and $p$ is not a winner in either of them.
  Thus, let us assume that there is at least one voter that prefers
  $p$ to all other candidates in
$C$. This means that we can assume, without
  loss of generality, that there is a candidate $d \in C$ such that
  for each candidate $c \in C - \{p,d\}$ it holds that $p
  \sord c \sord d$. By Lemma~\ref{l:s-c}, $p$ and $d$ are the only
  candidates in election $E$ whose plurality score is nonzero.
  We assume that $p$ does not win in $E$, so the score of $p$ is
  smaller than the score of $d$. We now consider three possible
  cases, depending on $a$'s position on the societal axis.
  \begin{enumerate}
  \item If $a \sord p$ then it is easy to note that in every vote in
    which $p$ was ranked first prior to adding $a$, now $a$ is ranked
    first, and so $p$ is not a winner of the election.
  \item If $d \sord a$ then $a$ is ranked first in each vote in which
    $d$ was ranked first prior to adding $a$, and so now $p$ loses to
    $a$.
  \item If $p \sord a \sord d$ then adding $a$ to the election does
    not change plurality scores of $p$ and $d$ and thus $p$ still
    loses to $d$.
  \end{enumerate}
  Thus, by induction on the number of added candidates, it is
  impossible to move $p$ from losing an election to winning it by
  adding candidates. Our CCAC algorithm simply checks if $p$ is a
  winner already, accepts if so and rejects otherwise.

  Let us now consider the case of CCDC for plurality and single-caved
  societies. Let $(C,V,p,K)$ be our input instance where votes in $V$
  are single-caved with respect to given societal axis $\sord$.  Let
  us rename the candidates so that $C = \{c_1, \ldots, c_m\}$, $c_1
  \sord \cdots \sord c_m$, and let us fix $i$ such that $p = c_i$.
  Let $\calD = \{ \{c_1, c_2 \ldots, c_{i-1}, c_j, \ldots, c_m\} \mid
  j > i\} \cup \{ \{c_1, \ldots, c_k,c_{i+1},c_{i+2}, \ldots, c_m\}
  \mid k < i\}$. By Lemma~\ref{l:s-c} and the definition of
  single-cavedness, it is easy to see that $p$ can become a winner of
  election $(C,V)$ by deleting at most $K$ candidates if and only if
  $V = \emptyset$ or
  there is a set $D$ in $\calD$ such that $\|D\| \leq K$ and $p$ is a
  winner of election $(C - D,V)$.
\end{proof}

\section{Proof from Section~\ref{sec:bribery}}

We provide the proof of 
Section~\ref{sec:bribery}'s theorem.

\newtheorem*{theorem5x1}{Theorem~\ref{t:approval-bribery}}

\begin{theorem5x1}
\bodyoffivepointone
\end{theorem5x1}

\begin{proof}
  We first prove in detail the ``bribery'' case (the first of the three
  types of bribery that the theorem covers), in both its marked-model
  and standard-model cases.  

  Let us look first at the marked model.  In this case, 
  much as in the proof of Theorem~\ref{t:stcp}, we will note that an
  earlier paper is implicitly obtaining a stronger result than what its
  theorem states, and then we will use that observation to build a
  disjunctive truth-table reduction from our problem to that problem.
In this case, the earlier paper is not
  \citet{fal-hem-hem-rot:j:single-peaked-preferences} as it was in
  Theorem~\ref{t:stcp}, but rather is
  \citet{bra-bri-hem-hem:c:sp2}.  The result of theirs that we focus
  on is their theorem stating that approval bribery is in polynomial time
  for single-peaked societies.  This is ``Theorem 4'' of
  \citet{bra-bri-hem-hem:c:sp2}, but for its proof/algorithm, one
  needs to refer to Appendix A.2 of that paper's 
  technical report version~\citep{bra-bri-hem-hem:t:sp2}.
  Now, by inspection of that proof, one can see that the algorithm
  given there does not need all the voters it operates on to respect the
  societal ordering.  Rather, it can handle perfectly well the case
  where each voter has an ``open to bribes'' flag, and every voter
  whose open to bribes flag is set respects the single-peaked
  ordering.  (In that algorithm, as modified to handle this, the
  surpluses 
  are computed with respect to all the voters---both those with
  the flag set and those with the flag unset.
  But then the pool of voters that the
  algorithm looks at to try to find a good bribe is limited to just
  those with the open to bribes flag set, although the surplus 
  recomputations throughout the algorithm are always with respect
  to the entire set of voters.  This is a slight extension of 
  the algorithm, but is clearly correct, for the same reasons
  the original algorithm is.)
  Note that the
  algorithm can bribe voters whose open to bribes flag is set, but
  (by the nature of the algorithm) will only bribe them to values
  consistent with the societal order.  Call the language problem 
  defined by this FlagBribe.

  Having made the previous paragraph's observation, we can now
  disjunctive truth-table reduce to FlagBribe---which is 
  put into polynomial time by the above algorithm (due to
  \citet{bra-bri-hem-hem:c:sp2}, except slightly adapted as 
  just mentioned).  We do so as follows.  Suppose our
  logarithmic bound is given.  Given an input to our problem, we check
  that the number of voters with the maverick-enabled flag (not to be
  confused with the open to bribes flag mentioned above) set
  does not exceed the logarithmic bound; if it does, reject
  immediately.  Otherwise, for each of member $A$ of the powerset of
  the set of maverick-enabled voters (i.e., for each choice of which
  of the maverick-enabled voters we will bribe), we will generate at
  most one instance of FlagBribe 
  as follows.  If $\|A\| > K$, generate no instance.
  (The number of voters being bribed would exceed the problem's bound.)
  Otherwise, generate an instance of FlagBribe that is the same set of voters as
  our instance, except with the members of $A$ modified to each
  approve of $p$ and only $p$.  The voters who in our original problem
  were not maverick-enabled will all have their open to bribes flag
  set.  All others will have their open to bribes flag unset.
  Set $K$ to now be $K - \|A\|$.

  So, since there are a logarithmic number of maverick-enabled voters,
  the powerset above is polynomial in size, and we generate a
  polynomial number of (polynomial-sized) instances of FlagBribe.  
  It is clear that our original problem has a successful bribe exactly
  if at least one of those instances has a successful bribe (i.e.,
  belongs to FlagBribe).  This is so, due to the properties of the
  algorithm underlying FlagBribe, and the fact that if there is a
  bribe of $K$ voters that makes $p$ a winner, then
the same bribe action except with any subset of them instead bribed to
approve only of $p$
  will also make $p$ a
  winner.  Changing a voter to approve of $p$ and only $p$ is a best
  possible bribe of that voter, if the voter will be bribed at all.

  That concludes the marked model case for bribery.  We turn now
  to the standard model case.  In this case, each voter can 
  potentially turn into a maverick.  And so with an
  $\bigo(\log(\mathrm{ProblemInputSize}))$ bound on the number of 
  mavericks, as long as the bribery problem itself has a 
  generously large $K$, to even decide which voters to make into
  mavericks would seem to involve 
  $ \|C\| \choose \bigo(\log(\mathrm{ProblemInputSize}))$
  options---superpolynomially many, which potentially is a worry
  if they can be turned into complex mavericks.  But we are 
  again saved here by the fact that any good bribe of a voter is 
  at least as good if one just bribes that voter to approve only $p$
  (and clearly that 
  vote is also inherently consistent with the societal ordering).
  That means that if there is a good set of bribes, then there is 
  a good set of bribes that never bribes people to vote in ways 
  that are inconsistent with the societal order.  But given that,
  we can turn this case into our marked-model case.
  (The calls to FlagBribe that will underlie the handling of that 
  case indeed also may present problems with
  superpolynomial numbers of options as to 
  which voters to bribe.  But due to the single-peakedness-respecting
  nature of all voters who are open to bribes, that can be handled
  easily---that is the real power of FlagBribe's underlying 
  algorithm: it uses single-peakedness
  to tame combinatorial explosion.)

  In particular, we can proceed here as follows.  Take our input.
  Reject if the number of voters who are inconsistent with 
  the societal ordering conflicts with our logarithmic bound.
  Otherwise, have the maverick-enabled flag be set for each 
  voter who violates the societal ordering and have the 
  maverick-enabled flag be unset for all other voters.  Keep the $K$
  parameter the same as it originally was.  And then solve
  this marked-model case as described above.  This works, due to the
  comments of the previous paragraph.

  That covers in detail the case of bribery.  The remaining two cases,
  negative-bribery and strongnegative-bribery, are similarly proven by
  noting that one can alter the algorithms from
  \citet{bra-bri-hem-hem:c:sp2} for those two cases, and by noting that if
  there is a good bribe in these models, then there is a good bribe
  where no bribed voter will have any approval set other than either
  ``just $p$'' or the empty set.
\end{proof}


\begin{thebibliography}{38}
\providecommand{\natexlab}[1]{#1}
\providecommand{\url}[1]{\texttt{#1}}
\expandafter\ifx\csname urlstyle\endcsname\relax
  \providecommand{\doi}[1]{doi: #1}\else
  \providecommand{\doi}{doi: \begingroup \urlstyle{rm}\Url}\fi

\bibitem[{{Bartholdi}} and Orlin(1991)]{bar-oli:j:polsci:strategic-voting}
J.~{{Bartholdi}}, III and J.~Orlin.
\newblock Single transferable vote resists strategic voting.
\newblock \emph{Social Choice and Welfare}, 8\penalty0 (4):\penalty0 341--354,
  1991.

\bibitem[{{Bartholdi}} et~al.(1989){{Bartholdi}}, Tovey, and
  Trick]{bar-tov-tri:j:manipulating}
J.~{{Bartholdi}}, III, C.~Tovey, and M.~Trick.
\newblock The computational difficulty of manipulating an election.
\newblock \emph{Social Choice and Welfare}, 6\penalty0 (3):\penalty0 227--241,
  1989.

\bibitem[{{Bartholdi}} et~al.(1992){{Bartholdi}}, Tovey, and
  Trick]{bar-tov-tri:j:control}
J.~{{Bartholdi}}, III, C.~Tovey, and M.~Trick.
\newblock How hard is it to control an election?
\newblock \emph{Mathematical and Computer Modeling}, 16\penalty0
  (8/9):\penalty0 27--40, 1992.

\bibitem[Brandt et~al.(2010{\natexlab{a}})Brandt, Brill, Hemaspaandra, and
  Hemaspaandra]{bra-bri-hem-hem:c:sp2}
F.~Brandt, M.~Brill, E.~Hemaspaandra, and L.~Hemaspaandra.
\newblock Bypassing combinatorial protections: Polynomial-time algorithms for
  single-peaked electorates.
\newblock In \emph{Proc.~of AAAI-10}, pages 715--722, July 2010{\natexlab{a}}.

\bibitem[Brandt et~al.(2010{\natexlab{b}})Brandt, Brill, Hemaspaandra, and
  Hemaspaandra]{bra-bri-hem-hem:t:sp2}
F.~Brandt, M.~Brill, E.~Hemaspaandra, and L.~Hemaspaandra.
\newblock Bypassing combinatorial protections: {Polynomial}-time algorithms for
  single-peaked electorates.
\newblock Technical Report TR-955, Department of Computer Science, University
  of Rochester, Rochester, NY, April 2010{\natexlab{b}}.

\bibitem[Buhrman and Homer(1992)]{buh-hom:c-short:superpolynomial}
H.~Buhrman and S.~Homer.
\newblock Superpolynomial circuits, almost sparse oracles, and the exponential
  hierarchy.
\newblock In \emph{Proc.~of FSTTCS-92}, pages 116--127. Springer-Verlag {\it
  Lecture Notes in Computer Science \#652}, December 1992.

\bibitem[Conitzer(2009)]{con:j:eliciting-singlepeaked}
V.~Conitzer.
\newblock Eliciting single-peaked preferences using comparison queries.
\newblock \emph{JAIR}, 35:\penalty0 161--191, 2009.

\bibitem[Conitzer et~al.(2007)Conitzer, Sandholm, and
  Lang]{con-lan-san:j:when-hard-to-manipulate}
V.~Conitzer, T.~Sandholm, and J.~Lang.
\newblock When are elections with few candidates hard to manipulate?
\newblock \emph{Journal of the ACM}, 54\penalty0 (3):\penalty0 Article~14,
  2007.

\bibitem[D{\'{\i}}az and Tor{\'{a}}n(1990)]{dia-tor:j:beta}
J.~D{\'{\i}}az and J.~Tor{\'{a}}n.
\newblock Classes of bounded nondeterminism.
\newblock \emph{Mathematical Systems Theory}, 23\penalty0 (1):\penalty0 21--32,
  1990.

\bibitem[Elkind et~al.(2010)Elkind, Faliszewski, and
  Slinko]{elk-fal-sli:c:votewise-dr}
E.~Elkind, P.~Faliszewski, and A.~Slinko.
\newblock On the role of distances in defining voting rules.
\newblock In \emph{Proc.~of AAMAS-10}, pages 375--382, 2010.

\bibitem[Erd\'{e}lyi and Rothe(2010)]{erd-rot:c:fallback}
G.~Erd\'{e}lyi and J.~Rothe.
\newblock Control complexity in fallback voting.
\newblock In \emph{Proc. of 16th Australasian Theory Symposium}, pages 39--48,
  January 2010.

\bibitem[Erd\'{e}lyi et~al.(2008)Erd\'{e}lyi, Hemaspaandra, Rothe, and
  Spakowski]{erd-hem-rot-spa:t2:lobbying-but-now-is-avg-and-juntas}
G.~Erd\'{e}lyi, L.~Hemaspaandra, J.~Rothe, and H.~Spakowski.
\newblock Frequency of correctness versus average-case polynomial time and
  generalized juntas.
\newblock Technical Report TR-934, Department of Computer Science, University
  of Rochester, Rochester, NY, June 2008.

\bibitem[Erd\'{e}lyi et~al.(2009{\natexlab{a}})Erd\'{e}lyi, Hemaspaandra,
  Rothe, and Spakowski]{erd-hem-rot-spa:j:junta}
G.~Erd\'{e}lyi, L.~Hemaspaandra, J.~Rothe, and H.~Spakowski.
\newblock Generalized juntas and {NP}-hard sets.
\newblock \emph{Theoretical Computer Science}, 410\penalty0 (38--40):\penalty0
  3995--4000, 2009{\natexlab{a}}.

\bibitem[Erd\'{e}lyi et~al.(2009{\natexlab{b}})Erd\'{e}lyi, Nowak, and
  Rothe]{erd-now-rot:j:sp-av}
G.~Erd\'{e}lyi, M.~Nowak, and J.~Rothe.
\newblock Sincere-strategy preference-based approval voting fully resists
  constructive control and broadly resists destructive control.
\newblock \emph{Mathematical Logic Quarterly}, 55\penalty0 (4):\penalty0
  425--443, 2009{\natexlab{b}}.

\bibitem[Erd\'{e}lyi et~al.(2010)Erd\'{e}lyi, Piras, and
  Rothe]{erd-pir-rot:t:bucklin}
G.~Erd\'{e}lyi, L.~Piras, and J.~Rothe.
\newblock Bucklin voting is broadly resistant to control.
\newblock Technical Report arXiv:1005.4115~[cs.GT], arXiv.org, May 2010.

\bibitem[Escoffier et~al.(2008)Escoffier, Lang, and
  {\"O}zt{\"u}rk]{esc-lan-ozt:c:single-peaked-consistency}
B.~Escoffier, J.~Lang, and M.~{\"O}zt{\"u}rk.
\newblock Single-peaked consistency and its complexity.
\newblock In \emph{Proc. of ECAI-08}, pages 366--370, July 2008.

\bibitem[Faliszewski et~al.(2009{\natexlab{a}})Faliszewski, Hemaspaandra, and
  Hemaspaandra]{fal-hem-hem:c:multimode}
P.~Faliszewski, E.~Hemaspaandra, and L.~Hemaspaandra.
\newblock Multimode attacks on elections.
\newblock In \emph{Proc.~of IJCAI-09}, pages 128--133, July 2009{\natexlab{a}}.

\bibitem[Faliszewski et~al.(2009{\natexlab{b}})Faliszewski, Hemaspaandra, and
  Hemaspaandra]{fal-hem-hem:j:bribery}
P.~Faliszewski, E.~Hemaspaandra, and L.~Hemaspaandra.
\newblock How hard is bribery in elections?
\newblock \emph{JAIR}, 35:\penalty0 485--532, 2009{\natexlab{b}}.

\bibitem[Faliszewski et~al.(2009{\natexlab{c}})Faliszewski, Hemaspaandra,
  Hemaspaandra, and Rothe]{fal-hem-hem-rot:b-too-short:richer}
P.~Faliszewski, E.~Hemaspaandra, L.~Hemaspaandra, and J.~Rothe.
\newblock A richer understanding of the complexity of election systems.
\newblock In \emph{Fundamental Problems in Computing}, pages 375--406.
  Springer, 2009{\natexlab{c}}.

\bibitem[Faliszewski et~al.(2009{\natexlab{d}})Faliszewski, Hemaspaandra,
  Hemaspaandra, and
  Rothe]{fal-hem-hem-rot:cOUTDATEDbyJOURNALwithPTR:single-peaked-preferences}
P.~Faliszewski, E.~Hemaspaandra, L.~Hemaspaandra, and J.~Rothe.
\newblock The shield that never was: {Societies} with single-peaked preferences
  are more open to manipulation and control.
\newblock In \emph{Proc.\ of TARK-09}, pages 118--127, July 2009{\natexlab{d}}.
\newblock Full version appears as
  \citep{fal-hem-hem-rot:j:single-peaked-preferences}.

\bibitem[Faliszewski et~al.(2009{\natexlab{e}})Faliszewski, Hemaspaandra,
  Hemaspaandra, and Rothe]{fal-hem-hem-rot:j:llull}
P.~Faliszewski, E.~Hemaspaandra, L.~Hemaspaandra, and J.~Rothe.
\newblock Llull and {Copeland} voting computationally resist bribery and
  constructive control.
\newblock \emph{JAIR}, 35:\penalty0 275--341, 2009{\natexlab{e}}.

\bibitem[Faliszewski et~al.(2010)Faliszewski, Hemaspaandra, and
  Hemaspaandra]{fal-hem-hem:j:cacm-survey}
P.~Faliszewski, E.~Hemaspaandra, and L.~Hemaspaandra.
\newblock Using complexity to protect elections.
\newblock \emph{Communications of the ACM}, 53\penalty0 (11):\penalty0 74--82,
  2010.

\bibitem[Faliszewski et~al.(2011)Faliszewski, Hemaspaandra, Hemaspaandra, and
  Rothe]{fal-hem-hem-rot:j:single-peaked-preferences}
P.~Faliszewski, E.~Hemaspaandra, L.~Hemaspaandra, and J.~Rothe.
\newblock The shield that never was: {Societies} with single-peaked preferences
  are more open to manipulation and control.
\newblock \emph{Information and Computation}, 209:\penalty0 89--107, 2011.

\bibitem[Friedgut et~al.(2008)Friedgut, Kalai, and
  Nisan]{fri-kal-nis:c:quantiative-gib-sat}
E.~Friedgut, G.~Kalai, and N.~Nisan.
\newblock Elections can be manipulated often.
\newblock In \emph{Proc.~of FOCS-08}, pages 243--249, October 2008.

\bibitem[Gailmard et~al.(2009)Gailmard, Patty, and
  Penn]{gai-pat-pen:b:arrow-on-single-peaked-domains}
S.~Gailmard, J.~Patty, and E.~Penn.
\newblock Arrow's theorem on single-peaked domains.
\newblock In E.~Aragon\'{e}s, C.~Bevi\'{a}, H.~Llavador, and N.~Schofield,
  editors, \emph{The Political Economy of Democracy}, pages 335--342.
  Fundaci\'{o}n BBVA, 2009.

\bibitem[Garey and Johnson(1979)]{gar-joh:b:int}
M.~Garey and D.~Johnson.
\newblock \emph{Computers and Intractability: {A} Guide to the Theory of
  {NP}-Completeness}.
\newblock {W. H. Freeman and Company}, 1979.

\bibitem[Goldsmith et~al.(1996)Goldsmith, Levy, and
  Mundhenk]{gol-lev-mun:j:limited-nondeterminism}
J.~Goldsmith, M.~Levy, and M.~Mundhenk.
\newblock Limited nondeterminism.
\newblock \emph{SIGACT News}, 27\penalty0 (2):\penalty0 20--29, 1996.

\bibitem[Hemaspaandra and Hemaspaandra(2007)]{hem-hem:j:dichotomy}
E.~Hemaspaandra and L.~Hemaspaandra.
\newblock Dichotomy for voting systems.
\newblock \emph{Journal of Computer and System Sciences}, 73\penalty0
  (1):\penalty0 73--83, 2007.

\bibitem[Hemaspaandra et~al.(2009)Hemaspaandra, Hemaspaandra, and
  Rothe]{hem-hem-rot:j:hybrid}
E.~Hemaspaandra, L.~Hemaspaandra, and J.~Rothe.
\newblock Hybrid elections broaden complexity-theoretic resistance to control.
\newblock \emph{Mathematical Logic Quarterly}, 55\penalty0 (4):\penalty0
  397--424, 2009.

\bibitem[Homer and Longpr\'{e}(1994)]{hom-lon:j:sparse}
S.~Homer and L.~Longpr\'{e}.
\newblock On reductions of {NP} sets to sparse sets.
\newblock \emph{Journal of Computer and System Sciences}, 48\penalty0
  (2):\penalty0 324--336, 1994.

\bibitem[Kintala and Fisher(1980)]{fis-kin:j:beta}
C.~Kintala and P.~Fisher.
\newblock Refining nondeterminism in relativized polynomial-time bounded
  computations.
\newblock \emph{SIAM Journal on Computing}, 9\penalty0 (1):\penalty0 46--53,
  1980.

\bibitem[Ladner et~al.(1975)Ladner, Lynch, and Selman]{lad-lyn-sel:j:com}
R.~Ladner, N.~Lynch, and A.~Selman.
\newblock A comparison of polynomial time reducibilities.
\newblock \emph{Theoretical Computer Science}, 1\penalty0 (2):\penalty0
  103--124, 1975.

\bibitem[Levin(1986)]{lev:j:average-case}
L.~Levin.
\newblock Average case complete problems.
\newblock \emph{SIAM Journal on Computing}, 15\penalty0 (1):\penalty0 285--286,
  1986.

\bibitem[Ogiwara and Watanabe(1991)]{ogi-wat:j:pbt}
M.~Ogiwara and O.~Watanabe.
\newblock On polynomial-time bounded truth-table reducibility of {NP} sets to
  sparse sets.
\newblock \emph{SIAM Journal on Computing}, 20\penalty0 (3):\penalty0 471--483,
  June 1991.

\bibitem[Procaccia and Rosenschein(2007)]{pro-ros:j:juntas}
A.~Procaccia and J.~Rosenschein.
\newblock Junta distributions and the average-case complexity of manipulating
  elections.
\newblock \emph{JAIR}, 28:\penalty0 157--181, 2007.

\bibitem[Sch{\"{o}}ning(1986)]{sch:j:closeness}
U.~Sch{\"{o}}ning.
\newblock Complete sets and closeness to complexity classes.
\newblock \emph{Mathematical Systems Theory}, 19\penalty0 (1):\penalty0 29--42,
  1986.

\bibitem[Walsh(2007)]{wal:c:uncertainty-in-preference-elicitation-aggregation}
T.~Walsh.
\newblock Uncertainty in preference elicitation and aggregation.
\newblock In \emph{Proc.~of AAAI-07}, pages 3--8, July 2007.

\bibitem[Walsh(2009)]{wal:c:where-hard-veto}
T.~Walsh.
\newblock Where are the really hard manipulation problems? {T}he phase
  transition in manipulating the veto rule.
\newblock In \emph{Proc.~of IJCAI-09}, pages 324--329, July 2009.

\end{thebibliography}
\end{document}